\documentclass[english,aps,superscriptaddress]{revtex4-1}
\usepackage[T1]{fontenc}
\usepackage[latin9]{inputenc}
\usepackage{color}
\usepackage{babel}
\usepackage{verbatim}
\usepackage{refstyle}
\usepackage{amsthm}
\usepackage{amsmath}
\usepackage{graphicx}
\usepackage[pdfstartview={FitH},bookmarks=false]
 {hyperref}

\theoremstyle{plain}
\newtheorem{thm}{\protect\theoremname}
  \theoremstyle{remark}
  \newtheorem{rem}[thm]{\protect\remarkname}

  \providecommand{\remarkname}{Remark}
\providecommand{\theoremname}{Theorem}

\makeatletter

\begin{document}

\title{The Asymptotic Cooling of Heat-Bath Algorithmic Cooling}

\author{Sadegh Raeisi}
\email[]{sraeisi@uwaterloo.ca}
\affiliation{Institute for Quantum Computing, University of Waterloo, Ontario, Canada}
\affiliation{Department of Physics and Astronomy, University of Waterloo, Ontario, Canada}
\author{Michele Mosca}
\affiliation{Institute for Quantum Computing, University of Waterloo, Ontario, Canada}
\affiliation{Department of Combinatorics and Optimization, University of Waterloo, Ontario, Canada}
\affiliation{Perimeter Institute for Theoretical Physics, Waterloo, Ontario, Canada}
\affiliation{Canadian Institute for Advanced Research, Toronto, Ontario, Canada}

\begin{abstract}
The purity of quantum states is a key requirement for many quantum
applications. Improving the purity is limited by fundamental
laws of thermodynamics.
Here we are probing the fundamental limits for a natural approach to this
problem, namely heat-bath algorithmic cooling(HBAC). % of cooling a physical qubit.
The existence of the cooling limit for HBAC techniques was proved by
Schulman et al.,
the limit however remained unknown for the past decade. Here for the first time we establish this limit. In the context of quantum thermodynamics, this corresponds
to the maximum extractable work from the quantum system.
We also establish, in the case of higher dimensional reset systems, how the performance of HBAC depends on the energy spectrum
of the reset system.
\end{abstract}

\maketitle
The purity of quantum states is often one of the limiting factors in many
applications and quantum technologies. For instance, the signal to noise
ratio (SNR) in spectroscopy and medical imaging \cite{ardenkjaer-larsen_increase_2003,hall_polarization-enhanced_1997,kurhanewicz2011analysis,nelson2013metabolic,ross2010hyperpolarized} or
the resolution in metrology and quantum sensing \cite{PhysRevLett.100.133601, PhysRevLett.107.113603, riedel2010atom,ye2008quantum}
are often limited by the purity of the quantum states.
High purity is also a necessity for quantum computation.
Fault-tolerant quantum computing relies on using fresh ancillary quantum bits.
Recently Ben-Or, Gottesman and Hassidim proposed a quantum refrigerator to prepare high purity
quantum states for this purpose using algorithmic cooling \cite{ben2013quantum}.

Different methods have been exploited to improve the purity but all of
these techniques are limited by the laws of thermodynamics \cite{horodecki2013fundamental, PhysRevE.84.041109}.
It is interesting both fundamentally and practically to understand these limits.
%These limits have been investigated from thermodynamic point of view.
In the context of quantum thermodynamics, extracting work from
a quantum system is equivalent to increasing its purity and cooling it
\cite{horodecki_reversible_2003} and cooling limits correspond to Carnot-like
efficiency limits. Quantum thermodynamics has been studied as a resource theory
of purity \cite{ horodecki_reversible_2003, gour2013resource, brandao_resource_2013}
and recently Horodecki and Oppenheim extended this paradigm for general
thermodynamic transformations.
They found the limit for the extractable work in terms of relative
entropy when the Hamiltonian of the process is time independent. Usually quantum applications
involve quantum control which means that the Hamiltonian is
time-dependent, and in these cases their result gives an upper bound.

Heat-bath algorithmic cooling is another method which takes a more practical
approach to the cooling problem. Here a natural subclass of general
thermodynamic transformations  is considered where we have
control over a part of the system, and have limited control over how
the system interacts with an external heat-bath \cite{PhysRevE.84.041109,
weimer_local_2008}. This model applies to a wide range of
quantum implementation techniques like nuclear magnetic resonance
(NMR) \cite{ ryan2008spin,baugh2005experimental,brassard_experimental_2014}, ion-traps
\cite{ion-trap-cooling} and recently in quantum optics \cite{xu_demon-like_2014}. % and superconducting qubits[***].
The HBAC methods have also been studied from the thermodynamic
viewpoint \cite{rempp_cyclic_2007,weimer_local_2008}.

%Schulman et al. proved that the achievable cooling in
%these techniques is physically limited \cite{schulman_physical_2005}. The limit however
%remained unknown for the past decade \cite{schulman_physical_2005,elias2007optimal,
%elias_semioptimal_2011,brassard_experimental_2014, brassard_prospects_2014}. Here for the first time we find
%this limit. In the context of quantum thermodynamics, this corresponds
%to the maximum extractable work from the quantum system.

%More specifically,
Here we consider a quantum system that is in interaction
with a heat-bath. The quantum system comprises two kind of qubits,
the computation qubits and the reset qubits. The computation qubits
are the high quality qubits with long decoherence time that are
used for computation. The reset qubits on the other hand have
shorter relaxation time and equilibrate fast. 
%We assume that the equilibrium state is $\rho_{eq} = e^{-\beta H} $, where $\beta =\frac{1}{K_b T}$ with $K_b$ the Boltzmann factor and $T$ the temperature.
%
%The heat-bath incorporates other degrees of freedom  in the environment
%that couple to the qubits in the quantum system.
%Usually, different qubits couple differently to these degrees of freedom.
%Some of them couple more strongly and are used as the reset qubits
%and some have a weaker coupling and better relaxation times
%and are therefore used for computation.
 Figure \ref{fig:model}
shows a schematic of the model that we are considering in our work.

This model applies to a variety of physical systems. For instance, in NMR,
the system is the few nuclear spins that can be controlled and the heat-bath 
comprises the other magnetic moments in the sample. These magnetic moments
couple to the nuclear spins in the system and eventually equilibrate them.
Different spin species have different coupling rates \cite{ryan2008spin,
baugh2005experimental}. %Usually, hydrogen
%nuclei couple more strongly to
%the heat-bath and have a shorter relaxation time. This makes it a good
%reset qubit and it has been used for this purpose in
%Similarly, for superconducting devices, a low Q resonator is used for the heat-bath[***].

\begin{figure}
\begin{centering}
\includegraphics[width=0.95\columnwidth]{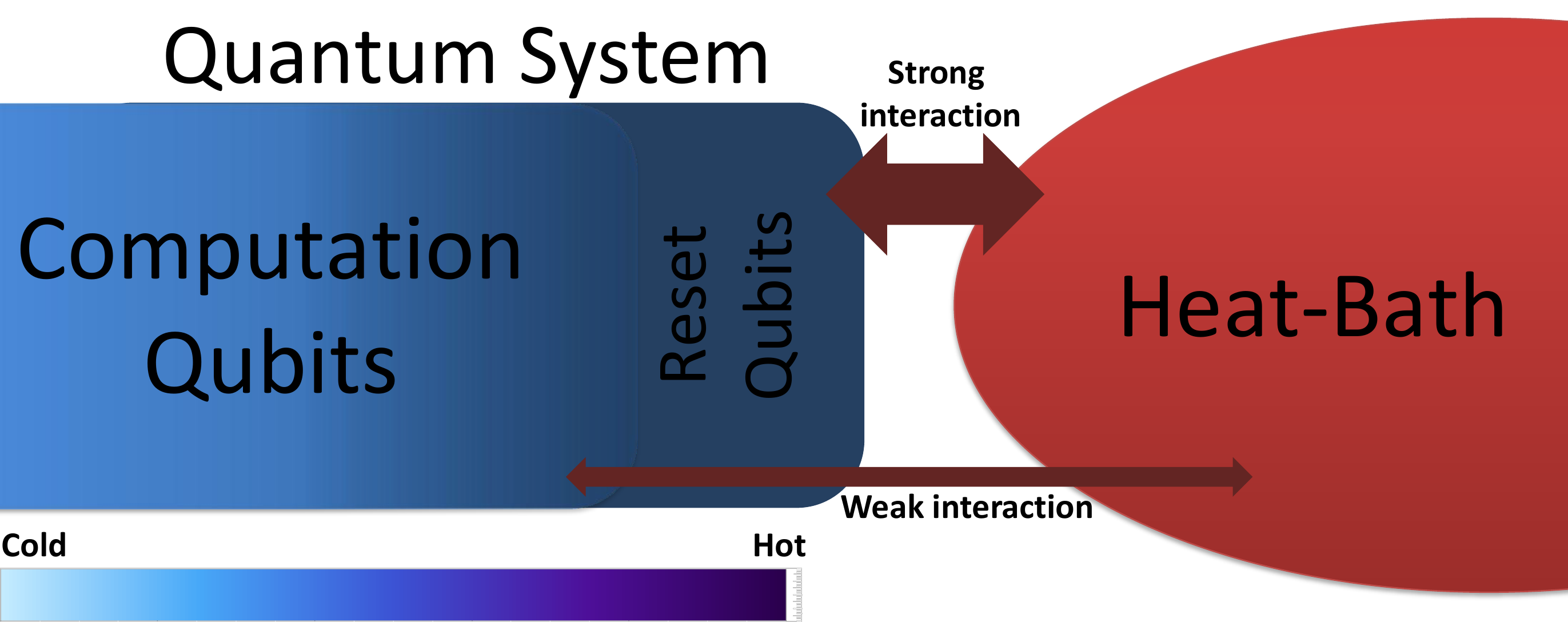}
\par\end{centering}
\caption{\label{fig:model} The schematic of the model. The quantum
system comprises computation qubits and reset qubits and interacts with a heat-bath.
The heat-bath incorporates degrees of freedom in the environment
that couple to the qubits in the quantum system.
Usually, different qubits couple differently to these degrees of freedom.
The computation qubits interact weakly and the reset qubits interact strongly with the heat-bath. We ignore the weak interaction between the computation qubits and the heat-bath and assume that only the reset qubits are effected by the interaction with the heat-bath.
The goal is to cool down the qubits in the system.
Note that this is just a schematic and in reality they are not necessarily spatially arranged in this way.
The HBAC does not cool all the qubits to the same temperature and
the asymptotic temperature of different computation qubits would be different.
We find the asymptotic state and consequently the temperature for all the 
qubits including the first one which is the cooling limit for all the 
HBAC techniques. }
\end{figure}

The class of cooling transformations that we are considering here
are known as heat-bath algorithmic cooling \cite{boykin_algorithmic_2002,
schulman_physical_2005,baugh2005experimental}. HBAC
is a quantum computation technique for  cooling
computation qubits by transferring their entropy to the reset qubits.
The reset qubits are regularly refreshed through their interaction with
the heat-bath.

The original idea of algorithmic cooling was developed by Schulman and
Vazirani in \cite{schulman_scalable_1998}
which uses a technique for Schumacher's quantum data compression
\cite{PhysRevA.54.2636,kaye2007cooling}. %The key ingredient of this method is the
%cooling for three qubits which is  called ``3BC''\cite{kaye2007cooling}. It gives
%a recipe to compress the entropy of three identical mixed qubits by
%transferring the entropy from the first qubit to the other two qubits.
%This effectively cools down the first qubit by a factor of $\frac{2}{3}$,
%while the other two qubits heat up and get more mixed.
%
Later it was proposed to use a heat-bath to enhance the cooling beyond
the compression bounds\cite{boykin_algorithmic_2002,fernandez2004algorithmic}.
The idea is that after the entropy transfer, the heat-bath refreshes
the hot qubits and then the entropy transfer can be repeated. Different
iterative methods were developed based on this idea\cite{elias2007optimal,
elias_semioptimal_2011,kaye2007cooling}.
All of these methods are  referred to as ``Heat-Bath Algorithmic Cooling''.

%MM this paragraph was very repetitive, I put part of it above,
%and shortened it
%As mentioned above, it is not possible to extract all the
%entropy of the computation qubits and cool them down arbitrarily \cite{schulman_physical_2005}.
%%%%%%%%%%%%%%%%%%%%%%%%%%%%%%%%Nov20
In \cite{schulman_physical_2005} Schulman et al. established a lower-bound for the asymptotic temperature
and proved that none of these iterative techniques can extract all
of the entropy from the computation qubits. However, the asymptotic cooling limit remained 
unknown. In \cite{elias2007optimal}, a steady state of HBAC was identified and was used 
to establish an upper-bound for the limit under the assumption that HBAC starts from 
the maximally mixed state and converges to a steady state.
%%%%%%%%%%%%%%%%%%%%%%%%%%%%%%%%OCT10
%Although a steady state of HBAC was identified and established a
%lower bound on the limit \cite{elias2007optimal}, proving that this was the actual HBAC limit (for an appropriate initial state) remained open for almost ten years. For other initial states, the limit may be different and also follows from our work.

In this work, we show that this process has
an asymptotic state and find this asymptotic
state of the computation and reset qubits.
This gives the cooling limit
 of the qubits in this framework. This fundamental limit sets the ultimate
limit of any practical cooling approach under similar constraints.

We use the technique that was introduced in \cite{schulman_physical_2005}.
It is called the ``Partner Pairing Algorithm (PPA)'' and is the
optimal technique for HBAC. We find the cooling limit for the PPA and as
it is the optimal technique, the limit applies to all the HBAC techniques as well.

The PPA is an iterative method. In each iteration,
the diagonal elements of the density matrix are sorted and then
the reset qubit is refreshed. 
For example, if we have $n=1$ computational qubits, plus one reset qubit, with combined probabilities corresponding to 0.45 for $|00\rangle$, 0.15 for $|01\rangle$, 0.3 for $|10\rangle$ and 0.1 for $|11\rangle$, then the sort step will swap $|01\rangle$ and $|10\rangle$. After this swap step, the probabilities of the computational basis states are in decreasing order with respect to the lexicographic ordering of the qubits, which corresponds to increasing the probability of a 0 in the leftmost qubit.
The reset process is equivalent to
\begin{equation}
R(\rho)=\textrm{Tr}_{R}(\rho)\otimes\rho_{R}.\label{eq:reset-process}
\end{equation}
$\textrm{Tr}_{R}\left(*\right)$ is the partial trace over the reset
qubit and $\rho_{R}=\frac{1}{e^{-\epsilon}+e^{\epsilon}}\left(\begin{array}{cc}
e^{\epsilon} & 0\\
0 & e^{-\epsilon}
\end{array}\right),\label{eq:reset-state}
$
is the fixed point of the reset process.% which is often the equilibrium state. 
The parameter $\epsilon$
is called the polarization and $\epsilon=\frac{\Delta}{2K_{b}T_{B}}$,
where $\Delta$ is the energy gap of the reset qubit, $K_{b}$ is
the Boltzmann constant, and $T_{B}$ is the temperature of the heat-bath. Polarization
is commonly used to quantify the purity of spins. The higher the polarization,
the purer and colder the qubit is. For a qubit with the state $\rho=\left(\begin{array}{cc}
a & 0\\
0 & b
\end{array}\right)$ the polarization is given by $\frac{1}{2}\log\left(\frac{a}{b}\right)$
.

The reset step cools down the reset qubit and changes the diagonal
elements of the density matrix which also changes their ordering.
The sort operation in the following iteration would then increase
the polarization of computation qubits. Figure \ref{fig:itration}
shows the procedure of each iteration. %

\begin{figure}
\begin{centering}
\includegraphics[width=0.95\columnwidth]{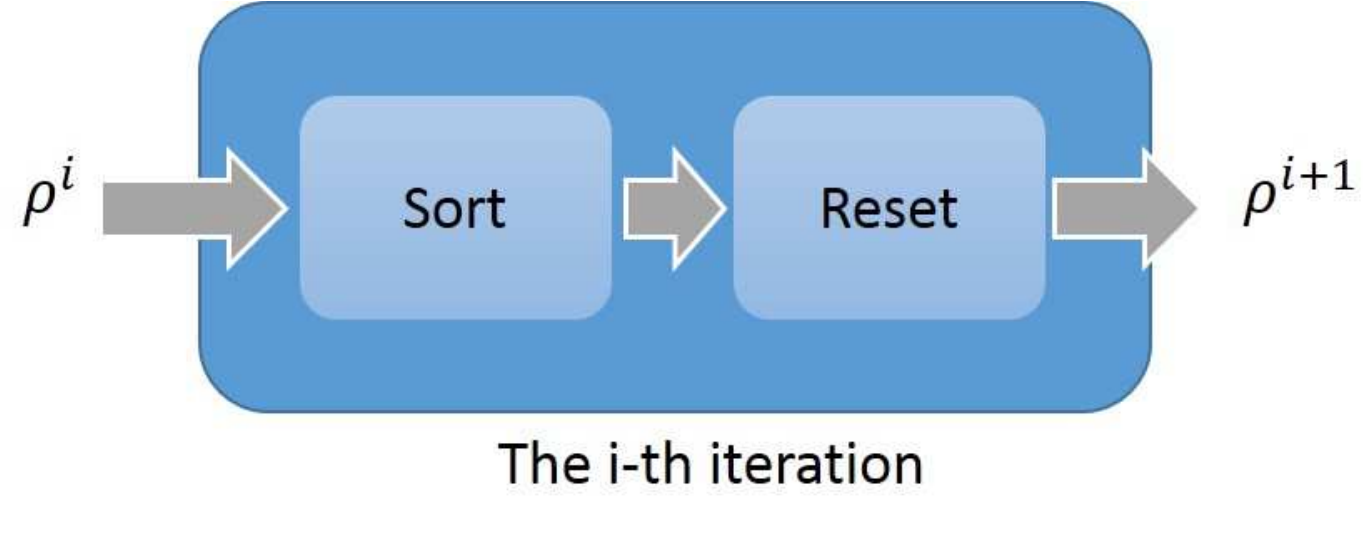}
\par\end{centering}
\caption{\label{fig:itration} The schematics of each iteration of PPA algorithmic
cooling. The diagonal elements of the density matrix are first sorted,
which increases the polarization of the first computation qubit and
decreases the polarization of the reset qubit. Next the reset process,
refreshes the reset qubit and restores its initial polarization.}
\end{figure}

%%%%%%%%%%%%%%%%%%%%%%%%%%%%%%%%%%%%%OCT10
We use $\left[\rho^{t}\right]=\left\{ \lambda_{1}^{(t)},\lambda_{2}^{(t)},\cdots\lambda_{2^{n+1}}^{(t)}\right\} \label{eq:diagonal_state}$
to show the state of $n$ computation qubits plus one reset qubit
which is the last one. The superscript represents the iteration index
and the subscript is the index of the diagonal elements. %

HBAC cools the first qubit monotonically which means that we just
need to find the asymptotic temperature to  find the cooling
limit. If the system converges to an asymptotic state (which
we show happens), then this state determines the cooling limit.

Despite the simplicity of each iteration, the dynamics are complicated
and it is difficult to understand how the state evolves under these
dynamics. In particular, even assuming the system converges to
an asymptotic state, it is challenging to find the asymptotic
state \cite{elias_semioptimal_2011,elias2007optimal}.
Note that the sort operation depends on the probability distribution 
and thus is changing in each iteration and
as a result, the cooling process is not a time-homogeneous
Markov process.

We use the fact that an asymptotic state should be invariant under PPA and first identify steady states of the PPA. This gives a necessary condition
for the asymptotic state. We then specify the asymptotic state by proving a
condition on the dynamics of PPA.

%We take a different approach for finding the cooling limit.
% and focus
%on the asymptotic iteration rather than the dynamics.
%MM this sentence doesn't sound right
% We determine when
%the dynamics stop and use this condition to find the asymptotic state.
%Instead of focusing on the dynamics, we note that the state of the system iteratively undergoing the HBAC algorithm will converge to an asymptotic state which is a fixed point of the HBAC procedure, and thus focus on the properties of this asymptotic state.

The asymptotic state does not change under the operations of HBAC and
%%%%%%%%%%%%%%%%%%%%%%%%%%%%%%%%%%%OCT10End
is a fixed point of the dynamics.
%Sept24  --> taking this back to July5 version, before the shortening
Technically this implies that %
if it is reset, it still will be sorted. The state after the reset is
$\left[\rho^{\infty}\right]=\left\{ p_{0}^{\infty},p_{1}^{\infty},\cdots p_{2^{n}-1}^{\infty}\right\} \otimes \left[\rho_{R}\right],\label{eq:Final_State}$
where the first part represents the state of the $n$ computation
qubits. %

The fact that the full density matrix is sorted after the reset step,
implies that $ p_{i}^{\infty}e^{-\epsilon}\geq p_{i+1}^{\infty}e^{\epsilon}\,,\forall i\label{eq:inf-cond}$.

%%%%%%%%%%%%%%%%%%%%%%%%%%%%%%%%%%%%%%%%%%%%%%%Oct08
Note that this condition does not specify the asymptotic state. In fact the steady state is not unique and any state that satisfies the condition above is invariant under PPA. Therefore the invariance under PPA is a necessary condition, but not sufficient. This was also recognized in \cite{elias2007optimal} where the
state in Equation (\ref{eq:steadystate}) 
was found and shown to establish a lower bound on the asymptotic
polarization. Their numerical evidence \cite{elias2007optimal}, and other 
numerical studies independent of this work \cite{NR2014}, suggested the 
bound is tight when the initial state is maximally mixed.
%Numerical evidence \cite{elias2007optimal,NR2014} suggests the bound is tight when the initial state is maximally mixed.

%%%%%%%%%%%%OCT22
%and the equality of the condition above has been identified as a lower bound for the asymptotic polarization. The equality condition was also rediscovered in \cite{NR2014}.
%%%%%%%%%%%%END:OCT22

One of the key elements of our work is the following theorem which specifies the steady state
%%%%%%%%%%%%%%%%%%%%%%%%%%%%%%%%%%%OCT10
that is  the asymptotic state of HBAC.
It states that while the distances between consecutive $p_i$ are increasing in PPA,
the ratio of two consecutive diagonal elements of the density matrix would never exceed $e^{2\epsilon}$.

%%%%%%%%%%%%%%%%%%%%%%%%%%%%%%%%%%%OCT14
\begin{thm} \label{Thm:Qubitmaxdist}
For PPA algorithmic cooling with a reset qubit $\left[\rho_{R}\right]=\frac{1}{e^{-\epsilon}+e^{\epsilon}}\left\{ e^{\epsilon},e^{-\epsilon}\right\} $, for any iteration $t$ and $i, 0 \leq i \leq 2^n-1$, $\frac{p_{i}^{t}}{p_{i+1}^{t}} \leq\max\left\{ e^{2\epsilon},\frac{p_{i}^{0}}{p_{i+1}^{0}}\right\} $.
\end{thm}
The sketch of the proof is as follows. For any index $i$ and any iteration
$t$, if the ratio of $\frac{p_{i}^{t}}{p_{i+1}^{t}}\leq e^{2\epsilon}$,
then we can show that $p_{i}^{t+1}\leq e^{\epsilon}\left(p_{i}^{t}+p_{i+1}^{t}\right)$
and $p_{i+1}^{t+1}\geq e^{-\epsilon}\left(p_{i}^{t}+p_{i+1}^{t}\right)$
and as a result $\frac{p_{i}^{t+1}}{p_{i+1}^{t+1}}\leq e^{2\epsilon}$.
On the other hand, if the ratio of $\frac{p_{i}^{t}}{p_{i+1}^{t}}\geq e^{2\epsilon}$,
then it is easy to see that $\frac{p_{i}^{t+1}}{p_{i+1}^{t+1}}\leq\frac{p_{i}^{t}}{p_{i+1}^{t}}$.
Note that the sort operation in this case could only decrease $p_i$ or
increase $p_{i+1}$, both of which leads to a lower $\frac{p_{i}^{t+1}}{p_{i+1}^{t+1}}$. Therefore we can always bound
$\frac{p_{i}^{t+1}}{p_{i+1}^{t+1}}\leq\max\left\{ e^{2\epsilon},\frac{p_{i}^{t}}{p_{i+1}^{t}}\right\} $.
Induction on $t$ completes the proof of the theorem. A more detailed proof
is given in the supplementary material.

%See the supplementary material for the proof of the theorem.

If the initial state satisfies $d^0_i \leq 2\epsilon$ for all $i$, 
which holds, for  a broad class of states like the maximally mixed state or the thermal
state when the computation qubits have a smaller gap than the reset qubit, 
then one obtains the following condition
%%%%%%%%%%%%%%%%%%%%%%%%%%%%%%%%end:oct8
for the asymptotic state:
\begin{equation}
p_{i}^{\infty}e^{-\epsilon}=p_{i+1}^{\infty}e^{\epsilon}\,,\forall i,\label{eq:cond}
\end{equation}
where $p_i$ are the diagonal elements of the density matrix of computation qubits.
%%%%%%%%%%%%%%%%%%%%%%%%%%%%%%%%%%%Oct29
Note that in general, it could be that $d^0_i \geq 2\epsilon$. We investigate the
more general case in the supplemental materials.
%%%%%%%%%%%%%%%%%%%%%%%%%%%%%%%%%%%%End:Oct29
%For a more detailed proof see the supplementary material.
%%%%%%%%%%%%%%%%%%%%%%%%%%%%%%%%%%%OCT10:END

This condition together with the normalization of the state is enough
to determine the full state. Equation (\ref{eq:cond}) can be rewritten
as $p_{i}^{\infty}=e^{-2i\epsilon}p_{0}^{\infty}$ and considering state
normalization gives:
\begin{equation}
p_{0}^{\infty}=\frac{e^{-2\epsilon}-1}{\left(e^{-2\epsilon}\right)^{2^{n}}-1}.\label{eq:P0}
\end{equation}

Schulman et al. upper bounded $\lambda_{1}^{\infty}$ by $\frac{e^{2^{n}\epsilon}}{2^{n}}$
in \cite{schulman_physical_2005} which is consistent with our result.
Note that $\lambda_1^{\infty} = \frac{e^{\epsilon}}{e^{-\epsilon}+e^{-\epsilon}}p_{0}^{\infty}$.
Figure \ref{fig:Comparison}
gives a comparison between this bound
and the actual value from equation (\ref{eq:P0}). Plots are for $n=2$
and one reset qubit.
Figure \ref{fig:Comparison}
 illustrates how the upper bound in \cite{schulman_physical_2005}
 gets looser as $\epsilon$ increases.

\begin{figure}
\begin{centering}
\includegraphics[width=0.95\columnwidth]{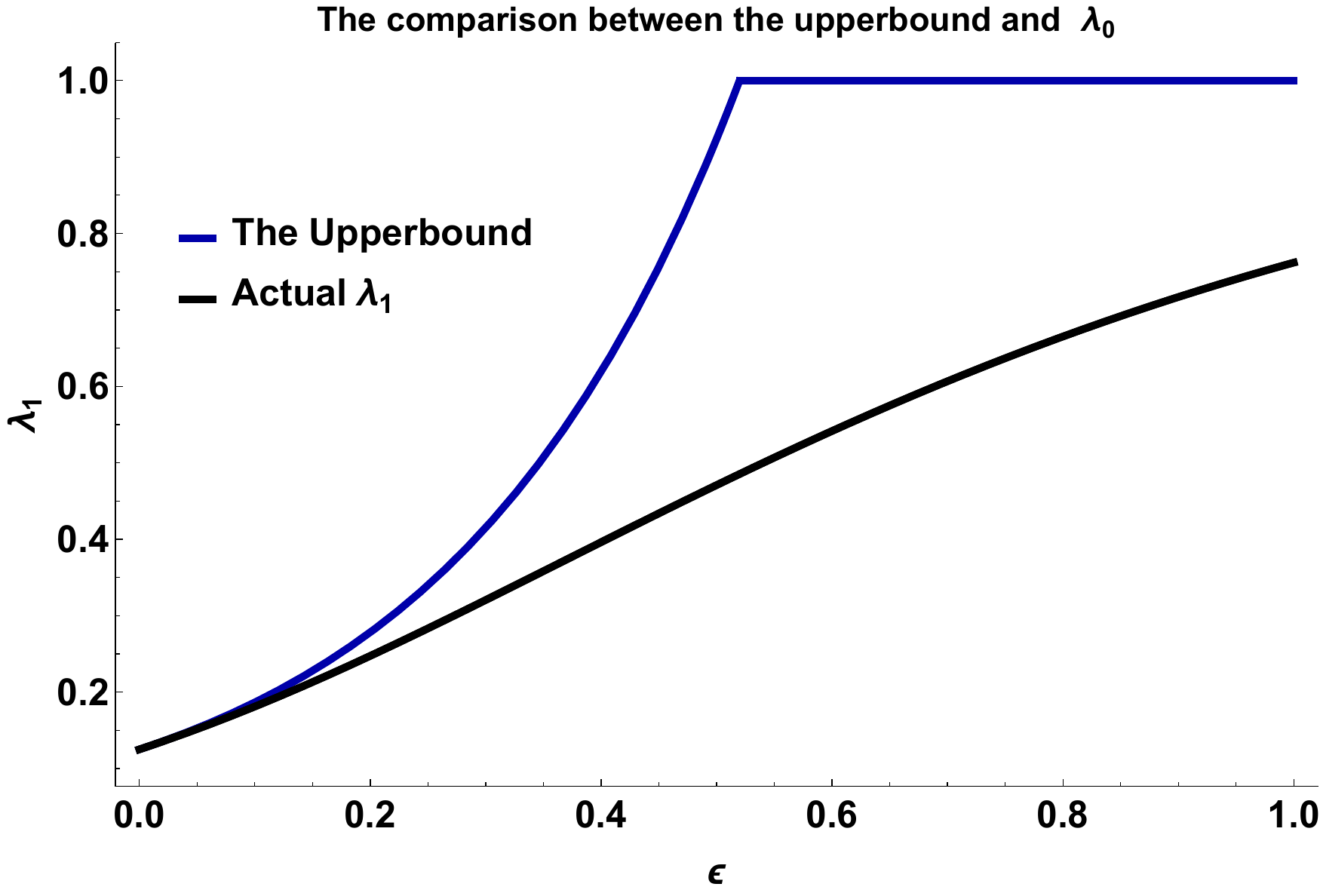}
\par\end{centering}
\caption{\label{fig:Comparison}Comparison of the upper bound and value of $\lambda_{1}^{\infty}$. %and the actual value from equation (\ref{eq:P0}).
The gap between the upper bound
and the actual value gets larger as $\epsilon$, the polarization
of the reset qubit increases. \footnote{Thanks to N. Rodrigues Briones and R. Laflamme for catching a glitch in an earlier version of this plot.}}
\end{figure}

Equations (\ref{eq:cond}) and (\ref{eq:P0}) give the asymptotic
state
\begin{equation}\label{eq:steadystate}
\left[\rho^{\infty}\right]=p_{0}^{\infty}\left\{ 1,\,
e^{-2\epsilon},\, e^{-4i\epsilon},\cdots\right\} \otimes\rho_{R}.
\end{equation} 

The first qubit has the lowest temperature. Therefore, we focus on
the first computation qubit for finding the cooling limit.
%The reduced density matrix of the first qubit can be calculated by tracing other qubits and
%as we mentioned before, that for the state $\rho=\left(\begin{array}{cc}
%a & 0\\
%0 & b
%\end{array}\right)$, the polarization is $\frac{1}{2}\log\left(\frac{a}{b}\right)$ which for the first qubit reduces to
We find that the polarization of the first qubit is
\begin{equation}
P=2^{n-1}\epsilon.\label{eq:Polarization}
\end{equation}

This result is consistent with the lower bound that was calculated
 in \cite{elias2007optimal}, in the case of  $\epsilon\ll\frac{1}{2^{n}}$.
%%%%%%%%%%%%%%%%%%%%%%%%%%%%%OCT10
In fact we proved that this lower-bound is tight.
%%%%%%%%%%%%%%%%%%%%%%%%%%%%%OCT10:END

Equation (\ref{eq:Polarization}) shows that the performance of HBAC
increases exponentially with the number of qubits, $n$. The simple
way to see this is to look at the effective temperature. The effective
temperature of the first qubit  is
%with the energy gap $\delta$ is given by $T_{\mbox{eff}}=\frac{\delta}{K_{b}}\frac{1}{\log\left(\frac{a}{b}\right)}$. For the first qubit it gives
\begin{equation}
T_{\mbox{eff}}=\frac{\delta}{\Delta}\frac{T_{B}}{2^{n-1}},\label{eq:Temperature}
\end{equation}
where $\delta$ is the energy gap of the qubit and is often
different from $\Delta$, the energy gap of the reset qubit. Usually
the reset and computation qubits should be of different species as
the reset qubits have a shorter relaxation time. The cooling limit
would improve if the energy gap
of the reset qubit is much larger than the one for the computation
qubits. For instance, if an electron is used as the reset qubit and
hydrogen nuclear spins for computation, this ratio would be $\frac{1}{660}$ which
lowers the cooling limit by a factor of $660$.

Figure \ref{fig:Tvsn}
shows how the effective temperature
decreases with increasing the number of computation qubits, $n$. It
also shows that changing the  $\frac{\delta}{\Delta}$ changes the cooling limit.

\begin{figure}
\begin{centering}
\includegraphics[width=0.95\columnwidth]{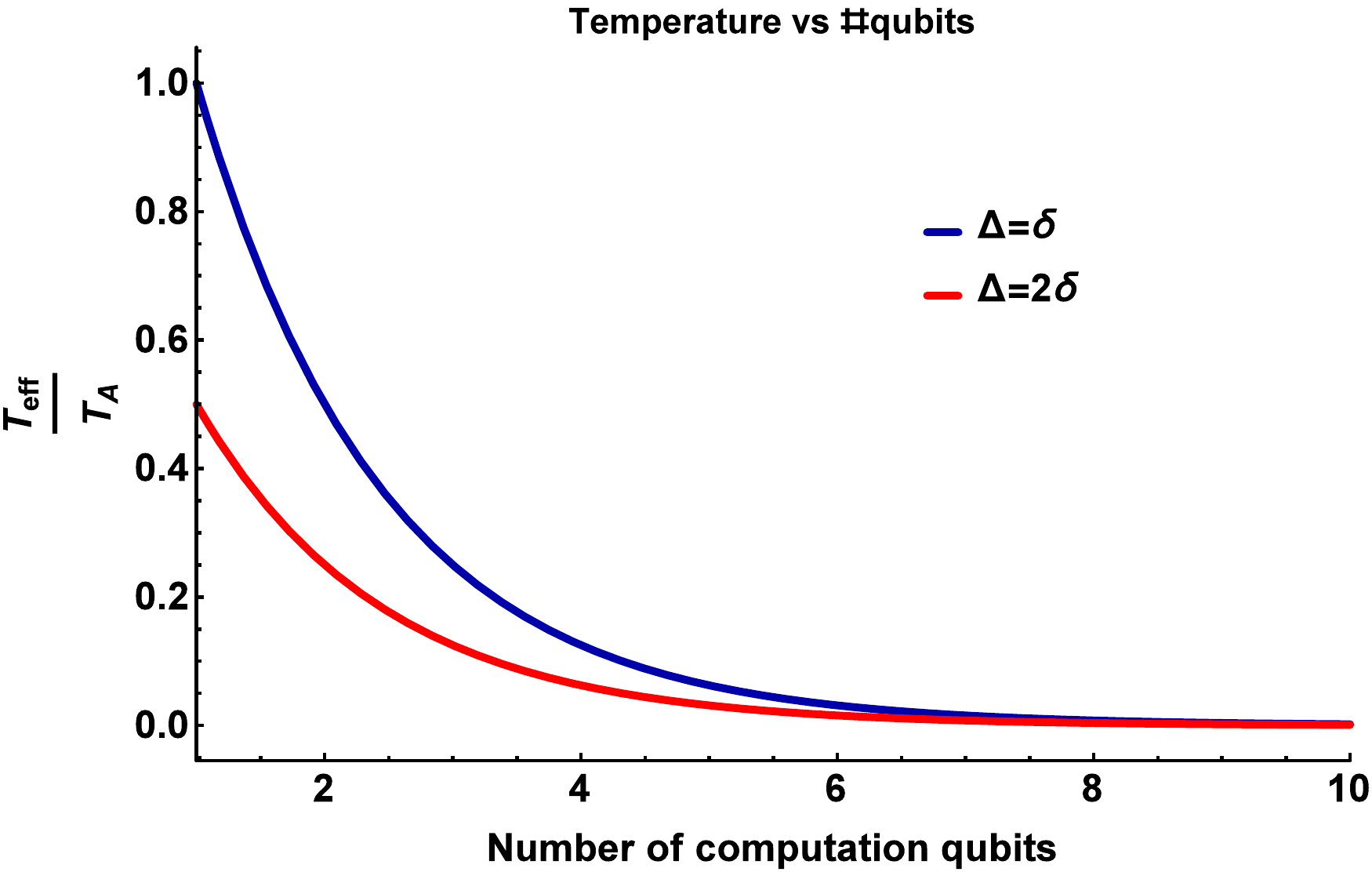}
\par\end{centering}

\caption{\label{fig:Tvsn} Asymptotic cooling ratio.
The cooling limit depends on the number of qubits and the ratio of
the energy gaps of the computation qubit to the one for the reset
qubit, $\frac{\delta}{\Delta}$. The cooling limit improves exponentially
with increasing n and linearly with decreasing $\frac{\delta}{\Delta}$. }
\end{figure}

We can also answer one of the important questions about HBAC, namely,
identifying how the performance of HBAC depends on the energy spectrum
of the reset system in the case of higher dimensional reset systems. %

Equation (\ref{eq:cond}) can be generalized for arbitrary reset state,
$\rho_{R}$. For a $D$-level reset system
%with the reset state $\left[\rho_{R}\right]=\frac{1}{z}\left\{ e^{\frac{-E_{1}}{K_{b}T}},e^{\frac{-E_{2}}{K_{b}T}},\cdots,e^{\frac{-E_{D}}{K_{b}T}}\right\} $,
we get a similar condition as in equation (\ref{eq:cond}) with the difference that the gap is replaced by the sum of the gaps.
% in the  $\Delta=E_{D}-E_{1}$.
%Consider a $D$-level reset system with eigenenergies
%$\left\{ E_{1},\, E_{2}\cdots,\, E_{D}\right\} $
%which gives the reset state
%$\left[\rho_{R}\right]=\frac{1}{z}\left\{ e^{\frac{-E_{1}}{K_{b}T}},e^{\frac{-E_{2}}{K_{b}T}},\cdots,e^{\frac{-E_{D}}{K_{b}T}}\right\} $. Again, we focus on the stopping condition
% which now changes to:
%\[
%p_{i}^{\infty}e^{\frac{-E_{D}}{K_{b}T}}\geq p_{i+1}^{\infty}e^{\frac{-E_{1}}{K_{b}T}}\,,\forall i.\]
%
%Similar to the qubit case, it can be shown  that the distance between
%the consecutive probabilities is bounded by $p_{i+1}^{\infty}e^{\frac{E_D-E_{1}}
%{K_{b}T}}$ (see the supplementary material) and therefore the equality holds.
%
%This is similar to the equality in equation (\ref{eq:cond}).
%The difference is that the gap is replaced by
%$\Delta=E_{D}-E_{1}$.
We refer to this as the ``large gap'' and
use $\Delta_{total}$ to show it. The cooling limit for qudits is 
\begin{equation}
T_{\mbox{eff}}=\frac{\delta}{\Delta_{total}}\frac{T_{B}}{2^{n-1}},\label{eq:Temperature-Qudit}
\end{equation}
 which similar to the one in 
Equation (\ref{eq:Temperature}) with $\Delta$ replaced by $\Delta_{total}$. 

It is interesting that despite the more
complicated energy structure of the reset qubit, the ``large gap''
is the only parameter that would influence the cooling limit.
In particular, the cooling limit does not directly depend on
the number of energy levels or the spacing between them, as long
as the total gap does not change.

%Specifically, it is interesting that the cooling limit does not directly
%depend on the number of energy levels of the reset system and a multi-level
%reset system would give the same cooling limit as a two level one,  if the
%large gap is the same.

This result also implies that a multi-qubit reset is linearly
better than a single qubit reset. The energy gap of a multi-qubit
reset system is the sum of the energy gaps of the individual qubits and
as a result it has a larger gap which would improve the cooling limit.
For instance, if $k$ identical qubits are used for the reset, then
the energy gap would be $\Delta_{total}=k\,\Delta$ and it lowers
the cooling limit by a factor of $\frac{1}{k}$. %

Note that the energy structure of the reset system could still change
the complexity or the number of operations for HBAC
% required to achieve a certain amount of cooling,
 but the asymptotic state only depends on the largest
gap of the reset system.

In conclusion, we establish the fundamental limit of cooling for all HBAC techniques
and show that it reduces exponentially with the number of qubits.
It also depends on the ratio of the energy gap of the reset qubit to
the gap of the computation qubits. We studied the effects of the changes
to the energy spectrum of the reset system and showed that only the
large gap of the reset system affects the asymptotic state. In particular,
the number of energy levels, for a constant energy gap, does not influence
the cooling limit.

%%%%%%%%%%%%%%%%%%%%%%Oct 24\\\\\\\\\\\\\\\\\\\\\
Note that experimental imperfections could affect the minimum
achievable cooling.
The HBAC operations are optimized based on the probability distribution at each step, and thus unknown errors in the probability distribution mean the cooling steps will not be optimal.
It is therefore critical to investigate
these imperfections. Some of these, like decoherence
has been recently studied \cite{brassard_prospects_2014}.

Besides the fundamental  significance, the
cooling limit could have practical applications as well. For instance,
it could give a quantitative measure of imperfection for implementing
and studying the HBAC.
One natural choice would be the distance from the
asymptotic state, $\rho^{\infty}$, which requires the full density
matrix. This may be expensive experimentally. An easier solution
is $\left|p_{0}-p_{0}^{\infty}\right|$  which approaches zero as the
state approaches the asymptotic cooling limit.  Or, simply $\left|P_0 - 2^{n-1}\epsilon \right|$, where $P_0$ is the achieved polarization of the first computation qubit.

For experimental implementation of HBAC, this measure
quantifies how far the experiment is from the cooling limit and gives
a metric for the assessment of the progress in the experiment.

Similarly, it can be used for theoretical cost analysis
of HBAC which requires a notion of distance from the asymptotic state.
The number of operations that are needed to achieve a certain fidelity
to the asymptotic state can be calculated in terms of such a distance.
Some studies have investigated this problem in terms of the number of iterations required \cite{schulman2007physical, schulman_physical_2005, NR2014}.
 By calculating the cost of an iteration in terms of resources such as gates or time, one could build on these works in order to calculate the resource cost of implementing HBAC.

%%%%%%%%%%%%%%%%%%%%%%%%%%%%%%%%%%%%%%End:OCT24

This work was supported by Canada's NSERC, MPrime, CIFAR, and CFI.
IQC and Perimeter Institute are supported in part by the Government of Canada and the Province of Ontario.

%After posting this paper we learned that equivalent forms of the fixed state  (eqn (1)) and polarization (in eq(2)) were discovered independently, after our discovery and before posting, by N. Briones Rodriguez and R. Laflamme in their work on HBAC.

%\bibliographystyle{aipauth4-1}
\bibliography{HBAC}

\begin{thebibliography}{33}%
\makeatletter
\providecommand \@ifxundefined [1]{%
 \@ifx{#1\undefined}
}%
\providecommand \@ifnum [1]{%
 \ifnum #1\expandafter \@firstoftwo
 \else \expandafter \@secondoftwo
 \fi
}%
\providecommand \@ifx [1]{%
 \ifx #1\expandafter \@firstoftwo
 \else \expandafter \@secondoftwo
 \fi
}%
\providecommand \natexlab [1]{#1}%
\providecommand \enquote  [1]{``#1''}%
\providecommand \bibnamefont  [1]{#1}%
\providecommand \bibfnamefont [1]{#1}%
\providecommand \citenamefont [1]{#1}%
\providecommand \href@noop [0]{\@secondoftwo}%
\providecommand \href [0]{\begingroup \@sanitize@url \@href}%
\providecommand \@href[1]{\@@startlink{#1}\@@href}%
\providecommand \@@href[1]{\endgroup#1\@@endlink}%
\providecommand \@sanitize@url [0]{\catcode `\\12\catcode `\$12\catcode
  `\&12\catcode `\#12\catcode `\^12\catcode `\_12\catcode `\%12\relax}%
\providecommand \@@startlink[1]{}%
\providecommand \@@endlink[0]{}%
\providecommand \url  [0]{\begingroup\@sanitize@url \@url }%
\providecommand \@url [1]{\endgroup\@href {#1}{\urlprefix }}%
\providecommand \urlprefix  [0]{URL }%
\providecommand \Eprint [0]{\href }%
\providecommand \doibase [0]{http://dx.doi.org/}%
\providecommand \selectlanguage [0]{\@gobble}%
\providecommand \bibinfo  [0]{\@secondoftwo}%
\providecommand \bibfield  [0]{\@secondoftwo}%
\providecommand \translation [1]{[#1]}%
\providecommand \BibitemOpen [0]{}%
\providecommand \bibitemStop [0]{}%
\providecommand \bibitemNoStop [0]{.\EOS\space}%
\providecommand \EOS [0]{\spacefactor3000\relax}%
\providecommand \BibitemShut  [1]{\csname bibitem#1\endcsname}%
\let\auto@bib@innerbib\@empty
%</preamble>
\bibitem [{\citenamefont {Ardenkj{\ae}r-Larsen}\ \emph
  {et~al.}(2003)\citenamefont {Ardenkj{\ae}r-Larsen}, \citenamefont {Fridlund},
  \citenamefont {Gram}, \citenamefont {Hansson}, \citenamefont {Hansson},
  \citenamefont {Lerche}, \citenamefont {Servin}, \citenamefont {Thaning},\
  and\ \citenamefont {Golman}}]{ardenkjaer-larsen_increase_2003}%
  \BibitemOpen
  \bibfield  {author} {\bibinfo {author} {\bibfnamefont {J.~H.}\ \bibnamefont
  {Ardenkj{\ae}r-Larsen}}, \bibinfo {author} {\bibfnamefont {B.}~\bibnamefont
  {Fridlund}}, \bibinfo {author} {\bibfnamefont {A.}~\bibnamefont {Gram}},
  \bibinfo {author} {\bibfnamefont {G.}~\bibnamefont {Hansson}}, \bibinfo
  {author} {\bibfnamefont {L.}~\bibnamefont {Hansson}}, \bibinfo {author}
  {\bibfnamefont {M.~H.}\ \bibnamefont {Lerche}}, \bibinfo {author}
  {\bibfnamefont {R.}~\bibnamefont {Servin}}, \bibinfo {author} {\bibfnamefont
  {M.}~\bibnamefont {Thaning}}, \ and\ \bibinfo {author} {\bibfnamefont
  {K.}~\bibnamefont {Golman}},\ }\href@noop {} {\bibfield  {journal} {\bibinfo
  {journal} {Proceedings of the National Academy of Sciences}\ }\textbf
  {\bibinfo {volume} {100}},\ \bibinfo {pages} {10158} (\bibinfo {year}
  {2003})}\BibitemShut {NoStop}%
\bibitem [{\citenamefont {Hall}\ \emph {et~al.}(1997)\citenamefont {Hall},
  \citenamefont {Maus}, \citenamefont {Gerfen}, \citenamefont {Inati},
  \citenamefont {Becerra}, \citenamefont {Dahlquist},\ and\ \citenamefont
  {Griffin}}]{hall_polarization-enhanced_1997}%
  \BibitemOpen
  \bibfield  {author} {\bibinfo {author} {\bibfnamefont {D.~A.}\ \bibnamefont
  {Hall}}, \bibinfo {author} {\bibfnamefont {D.~C.}\ \bibnamefont {Maus}},
  \bibinfo {author} {\bibfnamefont {G.~J.}\ \bibnamefont {Gerfen}}, \bibinfo
  {author} {\bibfnamefont {S.~J.}\ \bibnamefont {Inati}}, \bibinfo {author}
  {\bibfnamefont {L.~R.}\ \bibnamefont {Becerra}}, \bibinfo {author}
  {\bibfnamefont {F.~W.}\ \bibnamefont {Dahlquist}}, \ and\ \bibinfo {author}
  {\bibfnamefont {R.~G.}\ \bibnamefont {Griffin}},\ }\href {\doibase
  10.1126/science.276.5314.930} {\bibfield  {journal} {\bibinfo  {journal}
  {Science}\ }\textbf {\bibinfo {volume} {276}},\ \bibinfo {pages} {930}
  (\bibinfo {year} {1997})}\BibitemShut {NoStop}%
\bibitem [{\citenamefont {Kurhanewicz}\ \emph {et~al.}(2011)\citenamefont
  {Kurhanewicz}, \citenamefont {Vigneron}, \citenamefont {Brindle},
  \citenamefont {Chekmenev}, \citenamefont {Comment}, \citenamefont
  {Cunningham}, \citenamefont {DeBerardinis}, \citenamefont {Green},
  \citenamefont {Leach}, \citenamefont {Rajan} \emph
  {et~al.}}]{kurhanewicz2011analysis}%
  \BibitemOpen
  \bibfield  {author} {\bibinfo {author} {\bibfnamefont {J.}~\bibnamefont
  {Kurhanewicz}}, \bibinfo {author} {\bibfnamefont {D.~B.}\ \bibnamefont
  {Vigneron}}, \bibinfo {author} {\bibfnamefont {K.}~\bibnamefont {Brindle}},
  \bibinfo {author} {\bibfnamefont {E.~Y.}\ \bibnamefont {Chekmenev}}, \bibinfo
  {author} {\bibfnamefont {A.}~\bibnamefont {Comment}}, \bibinfo {author}
  {\bibfnamefont {C.~H.}\ \bibnamefont {Cunningham}}, \bibinfo {author}
  {\bibfnamefont {R.~J.}\ \bibnamefont {DeBerardinis}}, \bibinfo {author}
  {\bibfnamefont {G.~G.}\ \bibnamefont {Green}}, \bibinfo {author}
  {\bibfnamefont {M.~O.}\ \bibnamefont {Leach}}, \bibinfo {author}
  {\bibfnamefont {S.~S.}\ \bibnamefont {Rajan}},  \emph {et~al.},\ }\href@noop
  {} {\bibfield  {journal} {\bibinfo  {journal} {Neoplasia}\ }\textbf {\bibinfo
  {volume} {13}},\ \bibinfo {pages} {81} (\bibinfo {year} {2011})}\BibitemShut
  {NoStop}%
\bibitem [{\citenamefont {Nelson}\ \emph {et~al.}(2013)\citenamefont {Nelson},
  \citenamefont {Kurhanewicz}, \citenamefont {Vigneron}, \citenamefont
  {Larson}, \citenamefont {Harzstark}, \citenamefont {Ferrone}, \citenamefont
  {van Criekinge}, \citenamefont {Chang}, \citenamefont {Bok}, \citenamefont
  {Park}, \citenamefont {Reed}, \citenamefont {Carvajal}, \citenamefont
  {Small}, \citenamefont {Munster}, \citenamefont {Weinberg}, \citenamefont
  {Ardenkjaer-Larsen}, \citenamefont {Chen}, \citenamefont {Hurd},
  \citenamefont {Odegardstuen}, \citenamefont {Robb}, \citenamefont {Tropp},\
  and\ \citenamefont {Murray}}]{nelson2013metabolic}%
  \BibitemOpen
  \bibfield  {author} {\bibinfo {author} {\bibfnamefont {S.~J.}\ \bibnamefont
  {Nelson}}, \bibinfo {author} {\bibfnamefont {J.}~\bibnamefont {Kurhanewicz}},
  \bibinfo {author} {\bibfnamefont {D.~B.}\ \bibnamefont {Vigneron}}, \bibinfo
  {author} {\bibfnamefont {P.~E.~Z.}\ \bibnamefont {Larson}}, \bibinfo {author}
  {\bibfnamefont {A.~L.}\ \bibnamefont {Harzstark}}, \bibinfo {author}
  {\bibfnamefont {M.}~\bibnamefont {Ferrone}}, \bibinfo {author} {\bibfnamefont
  {M.}~\bibnamefont {van Criekinge}}, \bibinfo {author} {\bibfnamefont {J.~W.}\
  \bibnamefont {Chang}}, \bibinfo {author} {\bibfnamefont {R.}~\bibnamefont
  {Bok}}, \bibinfo {author} {\bibfnamefont {I.}~\bibnamefont {Park}}, \bibinfo
  {author} {\bibfnamefont {G.}~\bibnamefont {Reed}}, \bibinfo {author}
  {\bibfnamefont {L.}~\bibnamefont {Carvajal}}, \bibinfo {author}
  {\bibfnamefont {E.~J.}\ \bibnamefont {Small}}, \bibinfo {author}
  {\bibfnamefont {P.}~\bibnamefont {Munster}}, \bibinfo {author} {\bibfnamefont
  {V.~K.}\ \bibnamefont {Weinberg}}, \bibinfo {author} {\bibfnamefont {J.~H.}\
  \bibnamefont {Ardenkjaer-Larsen}}, \bibinfo {author} {\bibfnamefont {A.~P.}\
  \bibnamefont {Chen}}, \bibinfo {author} {\bibfnamefont {R.~E.}\ \bibnamefont
  {Hurd}}, \bibinfo {author} {\bibfnamefont {L.-I.}\ \bibnamefont
  {Odegardstuen}}, \bibinfo {author} {\bibfnamefont {F.~J.}\ \bibnamefont
  {Robb}}, \bibinfo {author} {\bibfnamefont {J.}~\bibnamefont {Tropp}}, \ and\
  \bibinfo {author} {\bibfnamefont {J.~A.}\ \bibnamefont {Murray}},\ }\href
  {\doibase 10.1126/scitranslmed.3006070} {\bibfield  {journal} {\bibinfo
  {journal} {Science Translational Medicine}\ }\textbf {\bibinfo {volume}
  {5}},\ \bibinfo {pages} {198ra108} (\bibinfo {year} {2013})},\ \Eprint
  {http://arxiv.org/abs/http://stm.sciencemag.org/content/5/198/198ra108.full.pdf}
  {http://stm.sciencemag.org/content/5/198/198ra108.full.pdf} \BibitemShut
  {NoStop}%
\bibitem [{\citenamefont {Ross}\ \emph {et~al.}(2010)\citenamefont {Ross},
  \citenamefont {Bhattacharya}, \citenamefont {Wagner}, \citenamefont {Tran},\
  and\ \citenamefont {Sailasuta}}]{ross2010hyperpolarized}%
  \BibitemOpen
  \bibfield  {author} {\bibinfo {author} {\bibfnamefont {B.}~\bibnamefont
  {Ross}}, \bibinfo {author} {\bibfnamefont {P.}~\bibnamefont {Bhattacharya}},
  \bibinfo {author} {\bibfnamefont {S.}~\bibnamefont {Wagner}}, \bibinfo
  {author} {\bibfnamefont {T.}~\bibnamefont {Tran}}, \ and\ \bibinfo {author}
  {\bibfnamefont {N.}~\bibnamefont {Sailasuta}},\ }\href@noop {} {\bibfield
  {journal} {\bibinfo  {journal} {American Journal of Neuroradiology}\ }\textbf
  {\bibinfo {volume} {31}},\ \bibinfo {pages} {24} (\bibinfo {year}
  {2010})}\BibitemShut {NoStop}%
\bibitem [{\citenamefont {Mosley}\ \emph {et~al.}(2008)\citenamefont {Mosley},
  \citenamefont {Lundeen}, \citenamefont {Smith}, \citenamefont {Wasylczyk},
  \citenamefont {U'Ren}, \citenamefont {Silberhorn},\ and\ \citenamefont
  {Walmsley}}]{PhysRevLett.100.133601}%
  \BibitemOpen
  \bibfield  {author} {\bibinfo {author} {\bibfnamefont {P.~J.}\ \bibnamefont
  {Mosley}}, \bibinfo {author} {\bibfnamefont {J.~S.}\ \bibnamefont {Lundeen}},
  \bibinfo {author} {\bibfnamefont {B.~J.}\ \bibnamefont {Smith}}, \bibinfo
  {author} {\bibfnamefont {P.}~\bibnamefont {Wasylczyk}}, \bibinfo {author}
  {\bibfnamefont {A.~B.}\ \bibnamefont {U'Ren}}, \bibinfo {author}
  {\bibfnamefont {C.}~\bibnamefont {Silberhorn}}, \ and\ \bibinfo {author}
  {\bibfnamefont {I.~A.}\ \bibnamefont {Walmsley}},\ }\href {\doibase
  10.1103/PhysRevLett.100.133601} {\bibfield  {journal} {\bibinfo  {journal}
  {Phys. Rev. Lett.}\ }\textbf {\bibinfo {volume} {100}},\ \bibinfo {pages}
  {133601} (\bibinfo {year} {2008})}\BibitemShut {NoStop}%
\bibitem [{\citenamefont {Thomas-Peter}\ \emph {et~al.}(2011)\citenamefont
  {Thomas-Peter}, \citenamefont {Smith}, \citenamefont {Datta}, \citenamefont
  {Zhang}, \citenamefont {Dorner},\ and\ \citenamefont
  {Walmsley}}]{PhysRevLett.107.113603}%
  \BibitemOpen
  \bibfield  {author} {\bibinfo {author} {\bibfnamefont {N.}~\bibnamefont
  {Thomas-Peter}}, \bibinfo {author} {\bibfnamefont {B.~J.}\ \bibnamefont
  {Smith}}, \bibinfo {author} {\bibfnamefont {A.}~\bibnamefont {Datta}},
  \bibinfo {author} {\bibfnamefont {L.}~\bibnamefont {Zhang}}, \bibinfo
  {author} {\bibfnamefont {U.}~\bibnamefont {Dorner}}, \ and\ \bibinfo {author}
  {\bibfnamefont {I.~A.}\ \bibnamefont {Walmsley}},\ }\href {\doibase
  10.1103/PhysRevLett.107.113603} {\bibfield  {journal} {\bibinfo  {journal}
  {Phys. Rev. Lett.}\ }\textbf {\bibinfo {volume} {107}},\ \bibinfo {pages}
  {113603} (\bibinfo {year} {2011})}\BibitemShut {NoStop}%
\bibitem [{\citenamefont {Riedel}\ \emph {et~al.}(2010)\citenamefont {Riedel},
  \citenamefont {B{\"o}hi}, \citenamefont {Li}, \citenamefont {H{\"a}nsch},
  \citenamefont {Sinatra},\ and\ \citenamefont {Treutlein}}]{riedel2010atom}%
  \BibitemOpen
  \bibfield  {author} {\bibinfo {author} {\bibfnamefont {M.~F.}\ \bibnamefont
  {Riedel}}, \bibinfo {author} {\bibfnamefont {P.}~\bibnamefont {B{\"o}hi}},
  \bibinfo {author} {\bibfnamefont {Y.}~\bibnamefont {Li}}, \bibinfo {author}
  {\bibfnamefont {T.~W.}\ \bibnamefont {H{\"a}nsch}}, \bibinfo {author}
  {\bibfnamefont {A.}~\bibnamefont {Sinatra}}, \ and\ \bibinfo {author}
  {\bibfnamefont {P.}~\bibnamefont {Treutlein}},\ }\href@noop {} {\bibfield
  {journal} {\bibinfo  {journal} {Nature}\ }\textbf {\bibinfo {volume} {464}},\
  \bibinfo {pages} {1170} (\bibinfo {year} {2010})}\BibitemShut {NoStop}%
\bibitem [{\citenamefont {Ye}\ \emph {et~al.}(2008)\citenamefont {Ye},
  \citenamefont {Kimble},\ and\ \citenamefont {Katori}}]{ye2008quantum}%
  \BibitemOpen
  \bibfield  {author} {\bibinfo {author} {\bibfnamefont {J.}~\bibnamefont
  {Ye}}, \bibinfo {author} {\bibfnamefont {H.}~\bibnamefont {Kimble}}, \ and\
  \bibinfo {author} {\bibfnamefont {H.}~\bibnamefont {Katori}},\ }\href@noop {}
  {\bibfield  {journal} {\bibinfo  {journal} {Science}\ }\textbf {\bibinfo
  {volume} {320}},\ \bibinfo {pages} {1734} (\bibinfo {year}
  {2008})}\BibitemShut {NoStop}%
\bibitem [{\citenamefont {Ben-Or}\ \emph {et~al.}(2013)\citenamefont {Ben-Or},
  \citenamefont {Gottesman},\ and\ \citenamefont {Hassidim}}]{ben2013quantum}%
  \BibitemOpen
  \bibfield  {author} {\bibinfo {author} {\bibfnamefont {M.}~\bibnamefont
  {Ben-Or}}, \bibinfo {author} {\bibfnamefont {D.}~\bibnamefont {Gottesman}}, \
  and\ \bibinfo {author} {\bibfnamefont {A.}~\bibnamefont {Hassidim}},\
  }\href@noop {} {\bibfield  {journal} {\bibinfo  {journal} {arXiv preprint
  arXiv:1301.1995}\ } (\bibinfo {year} {2013})}\BibitemShut {NoStop}%
\bibitem [{\citenamefont {Horodecki}\ and\ \citenamefont
  {Oppenheim}(2013)}]{horodecki2013fundamental}%
  \BibitemOpen
  \bibfield  {author} {\bibinfo {author} {\bibfnamefont {M.}~\bibnamefont
  {Horodecki}}\ and\ \bibinfo {author} {\bibfnamefont {J.}~\bibnamefont
  {Oppenheim}},\ }\href@noop {} {\bibfield  {journal} {\bibinfo  {journal}
  {Nature Communications}\ }\textbf {\bibinfo {volume} {4}} (\bibinfo {year}
  {2013})}\BibitemShut {NoStop}%
\bibitem [{\citenamefont {Allahverdyan}\ \emph {et~al.}(2011)\citenamefont
  {Allahverdyan}, \citenamefont {Hovhannisyan}, \citenamefont {Janzing},\ and\
  \citenamefont {Mahler}}]{PhysRevE.84.041109}%
  \BibitemOpen
  \bibfield  {author} {\bibinfo {author} {\bibfnamefont {A.~E.}\ \bibnamefont
  {Allahverdyan}}, \bibinfo {author} {\bibfnamefont {K.~V.}\ \bibnamefont
  {Hovhannisyan}}, \bibinfo {author} {\bibfnamefont {D.}~\bibnamefont
  {Janzing}}, \ and\ \bibinfo {author} {\bibfnamefont {G.}~\bibnamefont
  {Mahler}},\ }\href {\doibase 10.1103/PhysRevE.84.041109} {\bibfield
  {journal} {\bibinfo  {journal} {Phys. Rev. E}\ }\textbf {\bibinfo {volume}
  {84}},\ \bibinfo {pages} {041109} (\bibinfo {year} {2011})}\BibitemShut
  {NoStop}%
\bibitem [{\citenamefont {Horodecki}\ \emph {et~al.}(2003)\citenamefont
  {Horodecki}, \citenamefont {Horodecki},\ and\ \citenamefont
  {Oppenheim}}]{horodecki_reversible_2003}%
  \BibitemOpen
  \bibfield  {author} {\bibinfo {author} {\bibfnamefont {M.}~\bibnamefont
  {Horodecki}}, \bibinfo {author} {\bibfnamefont {P.}~\bibnamefont
  {Horodecki}}, \ and\ \bibinfo {author} {\bibfnamefont {J.}~\bibnamefont
  {Oppenheim}},\ }\href {\doibase 10.1103/PhysRevA.67.062104} {\bibfield
  {journal} {\bibinfo  {journal} {Phys. Rev. A}\ }\textbf {\bibinfo {volume}
  {67}},\ \bibinfo {pages} {062104} (\bibinfo {year} {2003})}\BibitemShut
  {NoStop}%
\bibitem [{\citenamefont {Gour}\ \emph {et~al.}(2013)\citenamefont {Gour},
  \citenamefont {M{\"u}ller}, \citenamefont {Narasimhachar}, \citenamefont
  {Spekkens},\ and\ \citenamefont {Halpern}}]{gour2013resource}%
  \BibitemOpen
  \bibfield  {author} {\bibinfo {author} {\bibfnamefont {G.}~\bibnamefont
  {Gour}}, \bibinfo {author} {\bibfnamefont {M.~P.}\ \bibnamefont
  {M{\"u}ller}}, \bibinfo {author} {\bibfnamefont {V.}~\bibnamefont
  {Narasimhachar}}, \bibinfo {author} {\bibfnamefont {R.~W.}\ \bibnamefont
  {Spekkens}}, \ and\ \bibinfo {author} {\bibfnamefont {N.~Y.}\ \bibnamefont
  {Halpern}},\ }\href@noop {} {\bibfield  {journal} {\bibinfo  {journal} {arXiv
  preprint arXiv:1309.6586}\ } (\bibinfo {year} {2013})}\BibitemShut {NoStop}%
\bibitem [{\citenamefont {Brand\~ao}\ \emph {et~al.}(2013)\citenamefont
  {Brand\~ao}, \citenamefont {Horodecki}, \citenamefont {Oppenheim},
  \citenamefont {Renes},\ and\ \citenamefont
  {Spekkens}}]{brandao_resource_2013}%
  \BibitemOpen
  \bibfield  {author} {\bibinfo {author} {\bibfnamefont {F.~G. S.~L.}\
  \bibnamefont {Brand\~ao}}, \bibinfo {author} {\bibfnamefont {M.}~\bibnamefont
  {Horodecki}}, \bibinfo {author} {\bibfnamefont {J.}~\bibnamefont
  {Oppenheim}}, \bibinfo {author} {\bibfnamefont {J.~M.}\ \bibnamefont
  {Renes}}, \ and\ \bibinfo {author} {\bibfnamefont {R.~W.}\ \bibnamefont
  {Spekkens}},\ }\href {\doibase 10.1103/PhysRevLett.111.250404} {\bibfield
  {journal} {\bibinfo  {journal} {Phys. Rev. Lett.}\ }\textbf {\bibinfo
  {volume} {111}},\ \bibinfo {pages} {250404} (\bibinfo {year}
  {2013})}\BibitemShut {NoStop}%
\bibitem [{\citenamefont {Weimer}\ \emph {et~al.}(2008)\citenamefont {Weimer},
  \citenamefont {Henrich}, \citenamefont {Rempp}, \citenamefont
  {Schr{\"o}der},\ and\ \citenamefont {Mahler}}]{weimer_local_2008}%
  \BibitemOpen
  \bibfield  {author} {\bibinfo {author} {\bibfnamefont {H.}~\bibnamefont
  {Weimer}}, \bibinfo {author} {\bibfnamefont {M.~J.}\ \bibnamefont {Henrich}},
  \bibinfo {author} {\bibfnamefont {F.}~\bibnamefont {Rempp}}, \bibinfo
  {author} {\bibfnamefont {H.}~\bibnamefont {Schr{\"o}der}}, \ and\ \bibinfo
  {author} {\bibfnamefont {G.}~\bibnamefont {Mahler}},\ }\href {\doibase
  10.1209/0295-5075/83/30008} {\bibfield  {journal} {\bibinfo  {journal} {{EPL}
  (Europhysics Letters)}\ }\textbf {\bibinfo {volume} {83}},\ \bibinfo {pages}
  {30008} (\bibinfo {year} {2008})}\BibitemShut {NoStop}%
\bibitem [{\citenamefont {Ryan}\ \emph {et~al.}(2008)\citenamefont {Ryan},
  \citenamefont {Moussa}, \citenamefont {Baugh},\ and\ \citenamefont
  {Laflamme}}]{ryan2008spin}%
  \BibitemOpen
  \bibfield  {author} {\bibinfo {author} {\bibfnamefont {C.~A.}\ \bibnamefont
  {Ryan}}, \bibinfo {author} {\bibfnamefont {O.}~\bibnamefont {Moussa}},
  \bibinfo {author} {\bibfnamefont {J.}~\bibnamefont {Baugh}}, \ and\ \bibinfo
  {author} {\bibfnamefont {R.}~\bibnamefont {Laflamme}},\ }\href {\doibase
  10.1103/PhysRevLett.100.140501} {\bibfield  {journal} {\bibinfo  {journal}
  {Phys. Rev. Lett.}\ }\textbf {\bibinfo {volume} {100}},\ \bibinfo {pages}
  {140501} (\bibinfo {year} {2008})}\BibitemShut {NoStop}%
\bibitem [{\citenamefont {Baugh}\ \emph {et~al.}(2005)\citenamefont {Baugh},
  \citenamefont {Moussa}, \citenamefont {Ryan}, \citenamefont {Nayak},\ and\
  \citenamefont {Laflamme}}]{baugh2005experimental}%
  \BibitemOpen
  \bibfield  {author} {\bibinfo {author} {\bibfnamefont {J.}~\bibnamefont
  {Baugh}}, \bibinfo {author} {\bibfnamefont {O.}~\bibnamefont {Moussa}},
  \bibinfo {author} {\bibfnamefont {C.~A.}\ \bibnamefont {Ryan}}, \bibinfo
  {author} {\bibfnamefont {A.}~\bibnamefont {Nayak}}, \ and\ \bibinfo {author}
  {\bibfnamefont {R.}~\bibnamefont {Laflamme}},\ }\href@noop {} {\bibfield
  {journal} {\bibinfo  {journal} {Nature}\ }\textbf {\bibinfo {volume} {438}},\
  \bibinfo {pages} {470} (\bibinfo {year} {2005})}\BibitemShut {NoStop}%
\bibitem [{\citenamefont {Brassard}\ \emph
  {et~al.}(2014{\natexlab{a}})\citenamefont {Brassard}, \citenamefont {Elias},
  \citenamefont {Fernandez}, \citenamefont {Gilboa}, \citenamefont {Jones},
  \citenamefont {Mor}, \citenamefont {Weinstein},\ and\ \citenamefont
  {Xiao}}]{brassard_experimental_2014}%
  \BibitemOpen
  \bibfield  {author} {\bibinfo {author} {\bibfnamefont {G.}~\bibnamefont
  {Brassard}}, \bibinfo {author} {\bibfnamefont {Y.}~\bibnamefont {Elias}},
  \bibinfo {author} {\bibfnamefont {J.~M.}\ \bibnamefont {Fernandez}}, \bibinfo
  {author} {\bibfnamefont {H.}~\bibnamefont {Gilboa}}, \bibinfo {author}
  {\bibfnamefont {J.~A.}\ \bibnamefont {Jones}}, \bibinfo {author}
  {\bibfnamefont {T.}~\bibnamefont {Mor}}, \bibinfo {author} {\bibfnamefont
  {Y.}~\bibnamefont {Weinstein}}, \ and\ \bibinfo {author} {\bibfnamefont
  {L.}~\bibnamefont {Xiao}},\ }\href@noop {} {\bibfield  {journal} {\bibinfo
  {journal} {arXiv preprint arXiv:1404.6885}\ } (\bibinfo {year}
  {2014}{\natexlab{a}})}\BibitemShut {NoStop}%
\bibitem [{\citenamefont {Barrett}\ \emph {et~al.}(2003)\citenamefont
  {Barrett}, \citenamefont {DeMarco}, \citenamefont {Schaetz}, \citenamefont
  {Meyer}, \citenamefont {Leibfried}, \citenamefont {Britton}, \citenamefont
  {Chiaverini}, \citenamefont {Itano}, \citenamefont
  {Jelenkovi\ifmmode~\acute{c}\else \'{c}\fi{}}, \citenamefont {Jost},
  \citenamefont {Langer}, \citenamefont {Rosenband},\ and\ \citenamefont
  {Wineland}}]{ion-trap-cooling}%
  \BibitemOpen
  \bibfield  {author} {\bibinfo {author} {\bibfnamefont {M.~D.}\ \bibnamefont
  {Barrett}}, \bibinfo {author} {\bibfnamefont {B.}~\bibnamefont {DeMarco}},
  \bibinfo {author} {\bibfnamefont {T.}~\bibnamefont {Schaetz}}, \bibinfo
  {author} {\bibfnamefont {V.}~\bibnamefont {Meyer}}, \bibinfo {author}
  {\bibfnamefont {D.}~\bibnamefont {Leibfried}}, \bibinfo {author}
  {\bibfnamefont {J.}~\bibnamefont {Britton}}, \bibinfo {author} {\bibfnamefont
  {J.}~\bibnamefont {Chiaverini}}, \bibinfo {author} {\bibfnamefont {W.~M.}\
  \bibnamefont {Itano}}, \bibinfo {author} {\bibfnamefont {B.}~\bibnamefont
  {Jelenkovi\ifmmode~\acute{c}\else \'{c}\fi{}}}, \bibinfo {author}
  {\bibfnamefont {J.~D.}\ \bibnamefont {Jost}}, \bibinfo {author}
  {\bibfnamefont {C.}~\bibnamefont {Langer}}, \bibinfo {author} {\bibfnamefont
  {T.}~\bibnamefont {Rosenband}}, \ and\ \bibinfo {author} {\bibfnamefont
  {D.~J.}\ \bibnamefont {Wineland}},\ }\href {\doibase
  10.1103/PhysRevA.68.042302} {\bibfield  {journal} {\bibinfo  {journal} {Phys.
  Rev. A}\ }\textbf {\bibinfo {volume} {68}},\ \bibinfo {pages} {042302}
  (\bibinfo {year} {2003})}\BibitemShut {NoStop}%
\bibitem [{\citenamefont {Xu}\ \emph {et~al.}(2014)\citenamefont {Xu},
  \citenamefont {Yung}, \citenamefont {Xu}, \citenamefont {Boixo},
  \citenamefont {Zhou}, \citenamefont {Li}, \citenamefont {Aspuru-Guzik},\ and\
  \citenamefont {Guo}}]{xu_demon-like_2014}%
  \BibitemOpen
  \bibfield  {author} {\bibinfo {author} {\bibfnamefont {J.-S.}\ \bibnamefont
  {Xu}}, \bibinfo {author} {\bibfnamefont {M.-H.}\ \bibnamefont {Yung}},
  \bibinfo {author} {\bibfnamefont {X.-Y.}\ \bibnamefont {Xu}}, \bibinfo
  {author} {\bibfnamefont {S.}~\bibnamefont {Boixo}}, \bibinfo {author}
  {\bibfnamefont {Z.-W.}\ \bibnamefont {Zhou}}, \bibinfo {author}
  {\bibfnamefont {C.-F.}\ \bibnamefont {Li}}, \bibinfo {author} {\bibfnamefont
  {A.}~\bibnamefont {Aspuru-Guzik}}, \ and\ \bibinfo {author} {\bibfnamefont
  {G.-C.}\ \bibnamefont {Guo}},\ }\href {\doibase 10.1038/nphoton.2013.354}
  {\bibfield  {journal} {\bibinfo  {journal} {Nature Photonics}\ }\textbf
  {\bibinfo {volume} {8}},\ \bibinfo {pages} {113} (\bibinfo {year}
  {2014})}\BibitemShut {NoStop}%
\bibitem [{\citenamefont {Rempp}\ \emph {et~al.}(2007)\citenamefont {Rempp},
  \citenamefont {Michel},\ and\ \citenamefont {Mahler}}]{rempp_cyclic_2007}%
  \BibitemOpen
  \bibfield  {author} {\bibinfo {author} {\bibfnamefont {F.}~\bibnamefont
  {Rempp}}, \bibinfo {author} {\bibfnamefont {M.}~\bibnamefont {Michel}}, \
  and\ \bibinfo {author} {\bibfnamefont {G.}~\bibnamefont {Mahler}},\ }\href
  {\doibase 10.1103/PhysRevA.76.032325} {\bibfield  {journal} {\bibinfo
  {journal} {Phys. Rev. A}\ }\textbf {\bibinfo {volume} {76}},\ \bibinfo
  {pages} {032325} (\bibinfo {year} {2007})}\BibitemShut {NoStop}%
\bibitem [{\citenamefont {Boykin}\ \emph {et~al.}(2002)\citenamefont {Boykin},
  \citenamefont {Mor}, \citenamefont {Roychowdhury}, \citenamefont {Vatan},\
  and\ \citenamefont {Vrijen}}]{boykin_algorithmic_2002}%
  \BibitemOpen
  \bibfield  {author} {\bibinfo {author} {\bibfnamefont {P.~O.}\ \bibnamefont
  {Boykin}}, \bibinfo {author} {\bibfnamefont {T.}~\bibnamefont {Mor}},
  \bibinfo {author} {\bibfnamefont {V.}~\bibnamefont {Roychowdhury}}, \bibinfo
  {author} {\bibfnamefont {F.}~\bibnamefont {Vatan}}, \ and\ \bibinfo {author}
  {\bibfnamefont {R.}~\bibnamefont {Vrijen}},\ }\href {\doibase
  10.1073/pnas.241641898} {\bibfield  {journal} {\bibinfo  {journal}
  {Proceedings of the National Academy of Sciences}\ }\textbf {\bibinfo
  {volume} {99}},\ \bibinfo {pages} {3388} (\bibinfo {year} {2002})},\ \bibinfo
  {note} {{PMID:} 11904402}\BibitemShut {NoStop}%
\bibitem [{\citenamefont {Schulman}\ \emph {et~al.}(2005)\citenamefont
  {Schulman}, \citenamefont {Mor},\ and\ \citenamefont
  {Weinstein}}]{schulman_physical_2005}%
  \BibitemOpen
  \bibfield  {author} {\bibinfo {author} {\bibfnamefont {L.~J.}\ \bibnamefont
  {Schulman}}, \bibinfo {author} {\bibfnamefont {T.}~\bibnamefont {Mor}}, \
  and\ \bibinfo {author} {\bibfnamefont {Y.}~\bibnamefont {Weinstein}},\ }\href
  {\doibase 10.1103/PhysRevLett.94.120501} {\bibfield  {journal} {\bibinfo
  {journal} {Phys. Rev. Lett.}\ }\textbf {\bibinfo {volume} {94}},\ \bibinfo
  {pages} {120501} (\bibinfo {year} {2005})}\BibitemShut {NoStop}%
\bibitem [{\citenamefont {Schulman}\ and\ \citenamefont
  {Vazirani}(1998)}]{schulman_scalable_1998}%
  \BibitemOpen
  \bibfield  {author} {\bibinfo {author} {\bibfnamefont {L.~J.}\ \bibnamefont
  {Schulman}}\ and\ \bibinfo {author} {\bibfnamefont {U.}~\bibnamefont
  {Vazirani}},\ }\href@noop {} {\bibfield  {journal} {\bibinfo  {journal}
  {arXiv preprint quant-ph/9804060}\ } (\bibinfo {year} {1998})}\BibitemShut
  {NoStop}%
\bibitem [{\citenamefont {Cleve}\ and\ \citenamefont
  {DiVincenzo}(1996)}]{PhysRevA.54.2636}%
  \BibitemOpen
  \bibfield  {author} {\bibinfo {author} {\bibfnamefont {R.}~\bibnamefont
  {Cleve}}\ and\ \bibinfo {author} {\bibfnamefont {D.~P.}\ \bibnamefont
  {DiVincenzo}},\ }\href {\doibase 10.1103/PhysRevA.54.2636} {\bibfield
  {journal} {\bibinfo  {journal} {Phys. Rev. A}\ }\textbf {\bibinfo {volume}
  {54}},\ \bibinfo {pages} {2636} (\bibinfo {year} {1996})}\BibitemShut
  {NoStop}%
\bibitem [{\citenamefont {Kaye}(2007)}]{kaye2007cooling}%
  \BibitemOpen
  \bibfield  {author} {\bibinfo {author} {\bibfnamefont {P.}~\bibnamefont
  {Kaye}},\ }\href@noop {} {\bibfield  {journal} {\bibinfo  {journal} {Quantum
  Information Processing}\ }\textbf {\bibinfo {volume} {6}},\ \bibinfo {pages}
  {295} (\bibinfo {year} {2007})}\BibitemShut {NoStop}%
\bibitem [{\citenamefont {Fernandez}\ \emph {et~al.}(2004)\citenamefont
  {Fernandez}, \citenamefont {Lloyd}, \citenamefont {Mor},\ and\ \citenamefont
  {Roychowdhury}}]{fernandez2004algorithmic}%
  \BibitemOpen
  \bibfield  {author} {\bibinfo {author} {\bibfnamefont {J.~M.}\ \bibnamefont
  {Fernandez}}, \bibinfo {author} {\bibfnamefont {S.}~\bibnamefont {Lloyd}},
  \bibinfo {author} {\bibfnamefont {T.}~\bibnamefont {Mor}}, \ and\ \bibinfo
  {author} {\bibfnamefont {V.}~\bibnamefont {Roychowdhury}},\ }\href@noop {}
  {\bibfield  {journal} {\bibinfo  {journal} {International Journal of Quantum
  Information}\ }\textbf {\bibinfo {volume} {2}},\ \bibinfo {pages} {461}
  (\bibinfo {year} {2004})}\BibitemShut {NoStop}%
\bibitem [{\citenamefont {Elias}\ \emph {et~al.}(2007)\citenamefont {Elias},
  \citenamefont {Fernandez}, \citenamefont {Mor},\ and\ \citenamefont
  {Weinstein}}]{elias2007optimal}%
  \BibitemOpen
  \bibfield  {author} {\bibinfo {author} {\bibfnamefont {Y.}~\bibnamefont
  {Elias}}, \bibinfo {author} {\bibfnamefont {J.~M.}\ \bibnamefont
  {Fernandez}}, \bibinfo {author} {\bibfnamefont {T.}~\bibnamefont {Mor}}, \
  and\ \bibinfo {author} {\bibfnamefont {Y.}~\bibnamefont {Weinstein}},\ }in\
  \href@noop {} {\emph {\bibinfo {booktitle} {Unconventional Computation}}}\
  (\bibinfo  {publisher} {Springer},\ \bibinfo {year} {2007})\ pp.\ \bibinfo
  {pages} {2--26}\BibitemShut {NoStop}%
\bibitem [{\citenamefont {Elias}\ \emph {et~al.}(2011)\citenamefont {Elias},
  \citenamefont {Mor},\ and\ \citenamefont
  {Weinstein}}]{elias_semioptimal_2011}%
  \BibitemOpen
  \bibfield  {author} {\bibinfo {author} {\bibfnamefont {Y.}~\bibnamefont
  {Elias}}, \bibinfo {author} {\bibfnamefont {T.}~\bibnamefont {Mor}}, \ and\
  \bibinfo {author} {\bibfnamefont {Y.}~\bibnamefont {Weinstein}},\ }\href
  {\doibase 10.1103/PhysRevA.83.042340} {\bibfield  {journal} {\bibinfo
  {journal} {Phys. Rev. A}\ }\textbf {\bibinfo {volume} {83}},\ \bibinfo
  {pages} {042340} (\bibinfo {year} {2011})}\BibitemShut {NoStop}%
\bibitem [{\citenamefont {Rodriguez~Briones}\ and\ \citenamefont
  {Laflamme}(2014)}]{NR2014}%
  \BibitemOpen
  \bibfield  {author} {\bibinfo {author} {\bibfnamefont {N.~A.}\ \bibnamefont
  {Rodriguez~Briones}}\ and\ \bibinfo {author} {\bibfnamefont {R.}~\bibnamefont
  {Laflamme}},\ }in\ \href@noop {} {\emph {\bibinfo {booktitle} {Poster
  presented at the Institute for Quantum Computing Scientific Advisory
  Committee poster session}}}\ (\bibinfo {year} {March 2014})\BibitemShut
  {NoStop}%
\bibitem [{\citenamefont {Brassard}\ \emph
  {et~al.}(2014{\natexlab{b}})\citenamefont {Brassard}, \citenamefont {Elias},
  \citenamefont {Mor},\ and\ \citenamefont
  {Weinstein}}]{brassard_prospects_2014}%
  \BibitemOpen
  \bibfield  {author} {\bibinfo {author} {\bibfnamefont {G.}~\bibnamefont
  {Brassard}}, \bibinfo {author} {\bibfnamefont {Y.}~\bibnamefont {Elias}},
  \bibinfo {author} {\bibfnamefont {T.}~\bibnamefont {Mor}}, \ and\ \bibinfo
  {author} {\bibfnamefont {Y.}~\bibnamefont {Weinstein}},\ }\href@noop {}
  {\bibfield  {journal} {\bibinfo  {journal} {arXiv preprint arXiv:1404.6824}\
  } (\bibinfo {year} {2014}{\natexlab{b}})}\BibitemShut {NoStop}%
\bibitem [{\citenamefont {Schulman}\ \emph {et~al.}(2007)\citenamefont
  {Schulman}, \citenamefont {Mor},\ and\ \citenamefont
  {Weinstein}}]{schulman2007physical}%
  \BibitemOpen
  \bibfield  {author} {\bibinfo {author} {\bibfnamefont {L.~J.}\ \bibnamefont
  {Schulman}}, \bibinfo {author} {\bibfnamefont {T.}~\bibnamefont {Mor}}, \
  and\ \bibinfo {author} {\bibfnamefont {Y.}~\bibnamefont {Weinstein}},\
  }\href@noop {} {\bibfield  {journal} {\bibinfo  {journal} {SIAM Journal on
  Computing}\ }\textbf {\bibinfo {volume} {36}},\ \bibinfo {pages} {1729}
  (\bibinfo {year} {2007})}\BibitemShut {NoStop}%
\end{thebibliography}%

\setcounter{equation}{0}
\setcounter{figure}{0}
\setcounter{table}{0}
\setcounter{page}{1}
\makeatletter
\renewcommand{\theequation}{S\arabic{equation}}
\renewcommand{\thefigure}{S\arabic{figure}}
\renewcommand{\bibnumfmt}[1]{[S#1]}
\renewcommand{\citenumfont}[1]{S#1}

%%%%%%%%%% Prefix a "S" to all equations, figures, tables and reset the counter %%%%%%%%%%

\section{Proof of the asymptotic state of PPA AC}

We first prove the Theorem \ref{Thm:Qubitmaxdist} 
in the paper and its
extension for qudit reset and then prove the convergence of
the PPA algorithmic cooling.
%Once we know that the asymptotic state exist, we can easily use the argument in equation (2) %**** to find the asymptotic state.

Before we get to the proof, it is useful to explain a few details
about the dynamics and the update rules of the PPA.

We start by explaining ``crossings'' which are the building block
of the cooling in PPA.

\subsection*{Crossings}

When the reset qubit is reset, the ordering of the elements on the
diagonal of the density matrix changes. These changes are what lead
to the cooling. The reset step takes the state $\left[\rho\right]
=\left\{ \lambda_{1},\lambda_{2},\cdots\lambda_{2^{n+1}}\right\}$ of $n$ computation
qubits and one reset qubit to
\[
\left[\rho^{'}\right]=\left\{ p_{0},p_{1},\cdots p_{2^{n}-1}\right\} \otimes\left[\rho_{R}\right],
\]
%MM we talk about a reset qubit then a reset qudit... the dimensions don't match
%%S is it acceptable now?

where $p_{i}=\lambda_{2i+1}+\lambda_{2i+2}$ and $\left[\rho_{R}\right]=\left\{e^\epsilon, e^{-\epsilon} \right\}$.
This can be generalized for the reset with a qudit $\left[\rho_{R}\right]=\left\{a_1, a_2,
 ..., a_k\right\} $ as well. Although the probabilities $p_{i}$
are sorted, the full density matrix is not necessarily sorted. For instance, for some indices $i < j$, and $m_i>m_j$, we could have
\begin{equation}
p_{i}a_{m_i} < p_{j}a_{m_j}, \label{eq:Crossingdown_cond}
\end{equation}
or similarly for some indices $i > j$, and $m_i < m_j$, we could get
\begin{equation}
p_{i}a_{m_i} > p_{j}a_{m_j}. \label{eq:Crossingtop_cond}
\end{equation}
%MM check the following sentence
In these cases, the sort operation in PPA would rearrange these %(and probably other)
 terms
and update the value of $p_{i}$. % and $p_{i+1}$.
We refer to the conditions in equation (\ref{eq:Crossingdown_cond}) and
equation (\ref{eq:Crossingtop_cond}) as ``crossing from below''
and ``crossing from above'' respectively.
%The sort operator would swap $p_i a_{m_i}$ with the appropriate term.
%For instance, when there is only one crossing from above, the updated
%$p_i$ would be
%\begin{equation}\label{eq:URup}
%p'_{i} =p_i - p_{i}a_{1} + p_{i-1}a_{k}.
%\end{equation}
%Similarly for crossing from below we get
%\begin{equation}\label{eq:URdown}
%p'_{i} =p_i - p_{i}a_{k} + p_{i+1}a_{1}.
%\end{equation}
%MM could you show me a simple example of this rule in action??
%not sure if it's worth including here, but I'd find it helpful
%For instance, consider $\left[\rho_{R}\right] = \{ ... p_i , p_{i+1} , ...\}\otimes \{a_1,a_2\}$, where $a_1 +a_2 =1$.
%If there is only one crossing between $p_i$ and $p_{i+1}$, namely $p_{i+1} a_1 > p_i a_2$, then the sort operation rearranges the state to $\left[\rho_{R}\right] = \{ ... p_i a_1,p_{i+1} a_1, p_i a_2 , p_{i+1} a_2 , ...\}$ which means that  $p'_i= (p_i+p_{i+1})a_1 = p_i - p_{i}a_{2} + p_{i+1}a_{1} $.

Despite the complexity of different crossings, we can make the following
general remarks.

\begin{rem}\label{Crossing-from-above}
\textbf{Crossing from above} If $p_i$ combined with the reset qudit probabilities gives the values
$\left\lbrace p_i a_1, p_i a_2, \cdots p_i a_k \right\rbrace $
and there is crossing from above (and none from below), the sort
operation on all $k2^n$ probabilities yields the values $ \lambda_1 \geq \lambda_2 \geq \cdots  \geq \lambda_k = p_i a_k$ for the $(ik+1)$th, $(ik+2)$th, \ldots $(ik+k)$th probabilities (i.e. the probabilities that will add up to determine $p_i^{\prime}$) such that
\begin{equation}
p_i a_j \geq \lambda_j \geq p_i a_k.
\end{equation}

\end{rem}

\begin{proof}
The second inequality is easier to see. It comes from the fact that
$\forall j<i, \, p_{j} \geq p_i$ and therefore, for any $m$, we have $p_{j} a_m \geq p_i a_k$.

The first inequality comes from the ordering. %If $p_i a_j < \lambda_j$, then  $p_i a_j$ should be the $\lambda_j$. In other words,
The value $\lambda_j$
is by definition the first element that is $ \geq \lambda_{j+1}$. One can use induction (starting with $k=j+1$ as the base case) to prove the first inequality. Specifically,
if $\lambda_{j+1} \leq p_i a_{j+1}$ then (since $a_j \geq a_{j+1}$) we have $p_i a_{j} \geq \lambda_{j+1}$ implies
which implies that $\lambda_j \leq p_i a_j$.

%so if by induction, it is always smaller than $p_i a_j < \delta_j$. (***put the induction).
\end{proof}

\begin{rem}
\textbf{\label{Crossing-from-below} Crossing from below} If $p_{i+1}$  combined with the reset qudit probabilities gives the values
$\left\lbrace p_{i+1} a_1, p_{i+1} a_2, \cdots p_{i+1} a_k \right\rbrace $
and there is crossing from below (and none from above), the sort
operation on all $k2^n$ probabilities yields the values
$\delta_1 = p_{i+1}a_1 \geq \delta_2 \geq \cdots  \geq \delta_k $
 for the $((i+1)k+1)$th, $((i+1)k+2)$th, \ldots $((i+1)k+k)$th probabilities such that
\begin{equation}
p_{i+1} a_j \leq \delta_j \leq p_{i+1} a_1.
\end{equation}

\end{rem}
The proof is similar to the one for crossing from above.

%Both of these two remarks follows directly from the update rule in equation (\ref{eq:URup}) and equation (\ref{eq:URdown}).

For simplicity we define a distance.
Consider two consecutive elements of the
density matrix, $p_{i}$ and $p_{i+1}$. We define the following distance
between the elements of the density matrix
\begin{equation}\label{Distance}
d_{i} \overset{\text{Def}}{=} \log\frac{p_{i}}{p_{i+1}}.
\end{equation}

We use the distance and generalise the Theorem
\ref{Thm:Qubitmaxdist} 
in the following way.

\begin{thm}\label{Thm:maxdist}
For PPA algorithmic cooling with a reset qudit $\left[\rho_{R}\right]=\left\{a_1, a_2, ..., a_k\right\} $ where $a_l$ are sorted decreasing and sum to $1$ and for any iteration $t$ , $d^t_i\leq \max \left\{ d^0_i, \log\frac{a_1}{a_k}\right\} $.
\end{thm}

\begin{proof}
%The proof is similar to the proof for the qubit case.
We focus on $d_i$ for some arbitrary iteration and prove the bound.

After the reset step of the iteration, we break down the sort operation
into two steps. First we separately sort the values $p_{j}a_l$ for $j\geq i+1$,
and also the values $p_ja_l$ for $j\leq i$. This means that the sort does not involve
sorting between terms of the form $p_{i} a_l$ and $p_{j\geq i+1} a_l$ nor between terms of the form $p_{i+1}a_l$
and $p_{j\leq i} a_l$.
%For this sort operation, $p_{i}$ only has crossings
%from above and $p_{i+1}$ has only crossings from below.
% which leads
%to
%\begin{align*}
%p_{i}=\sum_{j=1}^{k}p_{i}a_{j}\Rightarrow & p'_{i}=\sum_{j=1}^{k}\lambda_{j}\\
%p_{i+1}=\sum_{j=1}^{k}p_{i+1}a_{j}
%\Rightarrow & p'_{i+1}=\sum_{j=1}^{k}\delta_{j}.
%\end{align*}

Let $\lambda_{j}$ be the $(k-j+1)$th smallest number after the sort
operation among the $p_{j}a_l$ values for $j\leq i$ and let
$\delta_{j}$ be the
$j$th largest value after the sort operation among the $p_{j}a_l$ values for
$j\geq i+1$. Note that $\delta_1 = p_{i+1} a_1$ and $ \lambda_k = p_i a_k$.
In the next step we complete the sort by combining and sorting the probabilities that lie between $p_{i+1} a_1$ and $p_{i} a_k$.

Let's first consider the case that $\lambda_1 \geq \delta_1$ and
$\lambda_k \geq \delta_k$. 
Then when we merge and resort the $\lambda_j$ and $\delta_j$ values, for some integer $r \geq 0$ the values $\delta_1, \delta_2, \ldots, \delta_r$ will appear among the largest $k$ values, and $\delta_{r+1}, \ldots, \delta_{k}$ will appear among the smallest $k$ values. Similarly, the values $\lambda_1, \ldots, \lambda_{k-r}$ will appear among the largest $k$ values, and the values $\lambda_{k-r+1}, \ldots, \lambda_k$ will appear among the smallest $k$ values.

The sub-case where
$r=0$ corresponds to when there is no crossing between
$p_i$ and $p_{i+1}$  (i.e. $p_i a_k \geq p_{i+1} a_1$) and we will come back to this case as well.

Let us next consider the sub-case that $1 \leq r \leq k/2$.
Thus we get
\begin{align*}
& p'_{i}=\sum_{j=1}^{r}\left(\lambda_{j}+\delta_{j}\right)+\sum_{j=r+1}^{k-r}\lambda_{j}\\
 & p'_{i+1}=\sum_{j=k-r+1}^{k}\left(\lambda_{j}+\delta_{j}\right)+\sum_{j=r+1}^{k-r}\delta_{j}.
\end{align*}

% and $r \geq k$ corresponds to when either all the $\lambda_j \leq \delta_1$ or when all the $\delta_j \geq \lambda_k$ or both.

Using Remark \ref{Crossing-from-above} and Remark \ref{Crossing-from-below} we find the following bounds
\begin{align*}
p'_{i}\leq & \sum_{j=1}^{r}\left(p_{i}a_{j}+p_{i+1}a_{1}\right)+\sum_{j=r+1}^{k-r}p_{i}a_{j}\leq ra_{1}\left(p_{i}+p_{i+1}\right)+p_{i}\chi\\
p'_{i+1}\geq & \sum_{j=k-r+1}^{k}\left(p_{i}a_{k}+p_{i+1}a_{j}\right)+\sum_{j=r+1}^{k-r}p_{i+1}a_{j}\geq ra_{k}\left(p_{i}+p_{i+1}\right)+p_{i+1}\chi,
\end{align*}

where $\chi=\sum_{j=r+1}^{k-r}a_{j}$. We want to show that $\frac{p'_{i}}{p'_{i+1}}\leq\frac{a_{1}}{a_{k}}$ which follows from the fact that $p_{i}a_{k}\leq p_{i+1}a_{1}$ when there is at least one crossing between $p_{i}$ and $p_{i+1}$. In other words,

\begin{align*} \frac{p'_{i}}{p'_{i+1}}\leq\frac{ra_{1}\left(p_{i}+p_{i+1}\right)+p_{i}\chi}{ra_{k}\left(p_{i}+p_{i+1}\right)+p_{i+1}\chi}\leq\frac{ra_{1}\left(p_{i}+p_{i+1}\right)}{ra_{k}\left(p_{i}+p_{i+1}\right)}\leq\frac{a_{1}}{a_{k}}.
\end{align*}

Let us next consider the sub-case that $k/2 < r < k$.
Thus we get
\begin{align*}
\Rightarrow & p'_{i}=\sum_{j=1}^{k-r}\left(\lambda_{j}+\delta_{j}\right)+\sum_{j=k-r+1}^{r}\delta_{j}\\
\Rightarrow & p'_{i+1}=\sum_{j=r+1}^{k}\left(\lambda_{j}+\delta_{j}\right)+\sum_{j=k-r+1}^{r}\lambda_{j}.
\end{align*}

% and $r \geq k$ corresponds to when either all the $\lambda_j \leq \delta_1$ or when all the $\delta_j \geq \lambda_k$ or both.

Using Remark \ref{Crossing-from-above} and Remark \ref{Crossing-from-below} we find the following bounds
\begin{align*}
p'_{i}\leq & (k-r)(p_i + p_{i+1}) a_1 + (2r-k)p_{i+1}a_1 \\
p'_{i+1}\geq & (k-r)(p_i + p_{i+1})a_k + (2r-k)p_i a_k \\
\Rightarrow & p_{i}^\prime / p_{i+1}^\prime \leq 
\frac{\left((k-r)(p_i+p_{i+1}) + (2r-k) p_{i+1} \right)a_1}{\left((k-r)(p_i+p_{i+1}) + (2r-k) p_{i} \right)a_k} \\
& \leq 
\frac{\left((k-r)(p_i+p_{i+1}) + (2r-k) p_{i+1} \right)a_1}{\left((k-r)(p_i+p_{i+1}) + (2r-k) p_{i+1} \right)a_k} \leq \frac{a_1}{a_k}.
\\
\end{align*}

For the case where $r=0$, we get

\begin{align*}
p'_{i}=\sum_{j=1}^{k}\lambda_{j}\leq & \sum_{j=1}^{k}p_{i}a_{j}=p_{i}\\
p'_{i+1}=\sum_{j=1}^{k}\delta_{j}\geq & \sum_{j=1}^{k}p_{i+1}a_{j}=p_{i+1}.
\end{align*}

It follows that $\frac{p'_{i}}{p'_{i+1}}\leq\frac{p_{i}}{p{}_{i+1}}$.

Now we get to the case when either all
the $\lambda_j \leq \delta_1 = p_{i+1}a_1$ or when
all the $\delta_j \geq \lambda_k = p_i a_k$ (or both).
If it is the former, we get $p'_i \leq k \delta_1 = k p_{i+1} a_1$
and we also know that $p'_{i+1} \geq k p_{i+1} a_k$
which gives the desired result.
Similarly, if all the $\delta_j \geq \lambda_k$ then
$p'_{i+1} \geq k p_{i} a_k$ and $p'_i \leq k p_i a_1$
which again leads to $p'_i/p'_{i+1} \leq a_1/a_k$.
%If both condition applies, then $p'_i \leq k \delta_1 = k p_{i+1} a_1$
%and $p'_{i+1} \geq k p_{i} a_k$ which again implies that
%$p'_i/p'_{i+1} \leq a_1/a_k$ and concludes the proof.
%MM the proof of either case didn't seem to exclude both being true

%after sorting we
%get an ordered list of values of the following form:

%\begin{equation}
%p_{i+1} a_1 \geq  \zeta_1 \geq \zeta_2 \geq \cdots \geq p_i a_k.
%\end{equation}
%The important point here is that
%\begin{align*}
%p'_{i} \leq k p_{i+1}a_{1}\\
%p'_{i+1} \geq k p_{i}a_{k}.
%\end{align*}

%Since $p_{i+1} \leq p_i$, we get that
%$\frac{p'_i}{p'_{i+1}}\leq \frac{a_1}{a_k}$.

So the distance $d'_i$
is bounded above by $\max \left\lbrace  d_i,\frac{a_1}{a_k}\right\rbrace $ which proves the theorem.

%Clearly, any iteration of the PPA falls under a combination of these cases.

\end{proof}

Note that for the case of $k=2$ and $a_1  = e^\epsilon$ and $a_2 = e^{-\epsilon}$
we get Theorem \ref{Thm:Qubitmaxdist}.

%%%%%%%%%%%%%%%%%%%%%%%%%%%%%%%%%%%%%%%%%%%%%%%%%%END:Nov3

Now we use this to prove that all of the $p_i$ converge.

To prove the convergence, we first prove that $p_0$ converges and then the convergence
of all the $p_i$ follows from that. In order to make the connection between the
convergence of $p_0$ and other $p_i$, we use Theorem \ref{Thm:maxdist}.

\begin{thm}
Let $p_{0}^{t}$ be the first diagonal element of the reduced density matrix of the
computation qubits in the $t^{th}$ iteration of PPA algorithmic cooling. Then
$\underset{t\rightarrow\infty}{\lim}\, p_{0}^{t}=p^\infty_{0}$, for some constant $p^\infty_0$. \end{thm}
\begin{proof}
The sequence of values
$p_{0}^{t}$ are increasing because there can only be crossings from below for $p_0$. The sequence is also is upper-bounded, therefore it must converge:
$p^\infty_{0}=\lim_{t\rightarrow \infty} p_{0}^{t}$.
\end{proof}

\begin{thm}
Let $p_{i}^{t}$ be the $i^{th}$ diagonal element of the reduced density matrix of the
computation qubits in the $t^{th}$ iteration of PPA algorithmic cooling with a qubit reset $\left[\rho_{R}\right]=\frac{1}{e^{-\epsilon}+e^{\epsilon}}\left\{ e^{\epsilon},e^{-\epsilon}\right\} $. Then assuming
that $d_i^0\leq 2\epsilon, \forall i$, the limit
$\underset{t\rightarrow\infty}{\lim}\, p_{i}^{t} = e^{-2 i \epsilon}p^\infty_0 $ exists. We refer to the limit as
$p_{i}^{\infty}$.
\label{Thm:convergence}
\end{thm}
\begin{proof}
We already proved that the $p_0^{\infty}$ exists. This means that% for any $\delta^j > 0$, there is a $j$ such that
\begin{equation}
\underset{t\rightarrow\infty}{\lim}\, (p_0^{t+1}-p_0^{t}) = 0.
\end{equation}
On the other hand, Theorem \ref{Thm:maxdist} implies that as $t \rightarrow \infty$
\begin{equation}
p_0^{t+1}-p_0^{t} = p_1^{t} \frac{e^\epsilon }{z} - p_0^{t}\frac{e^{-\epsilon} }{z}.
\label{deltaP}\end{equation}

The limit of the last term, $p_0^{t}\frac{e^{-\epsilon} }{z}$ exists and the left hand side converges to zero, so $\underset{j\rightarrow\infty}{\lim}\, p_1^{t} $ must be $e^{-2 \epsilon}p^\infty_0$.
%%%%%%%%%%%%%%%%%%%%%%%%%%%%%%%OCT27
The convergence of the rest of the $p_{i}^{t}$ follows by induction.
Note that although there could be crossings from above for $i\geq 1$,
the change from above approaches zero and we get
\begin{equation}
p_i^{t+1}-p_i^{t} =\zeta + p_{i+1}^{t} \frac{e^\epsilon }{z} - p_i^{t}\frac{e^{-\epsilon} }{z},
\end{equation}
where $\zeta$ accounts for the changes to $p_i$ from crossings from above
and $\lim_{t\rightarrow\infty}\zeta=0$. Therefore we get
\begin{equation}
p_{i+1}^{\infty} = e^{-2\epsilon}p_i^{\infty}.
\end{equation}
%%%%%%%%%%%%%%%%%%%%%%%%%%End:Oct27
\end{proof}

For the PPA with a reset qudit $\left[\rho_{R}\right]=\left\{a_1, a_2, ..., a_k\right\} $,
 the proof is similar. As $j \rightarrow \infty$ we get $(p_0^{t+1}-p_0^{t}) - (p_1^{t}
 a_1 - p_0^{t}a_k) \rightarrow 0$  and since as $t \rightarrow \infty$ the left hand
 side and the last term converge, so does $p_1^{t}$. The rest of the proof follows similarly.

 If we start with a maximally mixed state for the computation qubits, $d^0_i=0$, i.e. initially all the distances are zero, then for any iteration $t$ and any index $i$, we get
\begin{equation}
d^t_i \leq  \log\left(\frac{a_{1}}{a_{k}}\right)
\end{equation}
and thus Theorem \ref{Thm:convergence} (and its generalization to qudits) applies.

Note that for the case where the qubits are not initially
in the maximally mixed state, Theorem \ref{Thm:maxdist} still applies and can be used to find the asymptotic state.
A sufficient condition for getting the asymptotic state in
Equation (\ref{eq:steadystate})
is that $\log_2(\frac{p^0_0}{p^0_{2^n-1}})\leq 2(2^n-1)\epsilon$.

For the more general case of
$\left[\rho^{0}\right]=\left\{ p^0_{0},p^0_{1},\cdots p^0_{2^{n}-1}\right\} $,
it is more complicated to determine the asymptotic state, however, Theorem \ref{Thm:maxdist} applies. In this case, the probabilities could be grouped in
different blocks of consecutive probabilities where in each block, the distance between any two
consecutive $p^0_i$ is less than $2\epsilon$ and is greater between two different
blocks.
Using the theorem, we can see that the distance between the
probabilities in each block would increase to $2\epsilon$. This
also implies that two neighbouring blocks may merge together. To
find the asymptotic state, we can go through the expansion and merger
of all the blocks until the final asymptotic state is found.
The asymptotic state would be a combination of different blocks where
$d^\infty_i  = 2\epsilon$ inside the blocks and is greater than that between
the blocks.

\end{document}